\documentclass[11pt,aps,pra,notitlepage,nofootinbib,tightenlines]{revtex4-1}

\usepackage[T1]{fontenc}
\usepackage[english]{babel}
\usepackage[utf8]{inputenc}

\usepackage[dvipsnames]{xcolor}
\usepackage{amsmath}
\usepackage{amssymb}
\usepackage{amsthm}
\usepackage{mathtools}
\usepackage{tikz}
\usepackage[colorlinks=true,linkcolor=blue,citecolor=blue]{hyperref}
\usepackage{enumitem}
\setlist[enumerate]{noitemsep,partopsep=0pt,parsep=0pt}
\setlist[itemize]{noitemsep,partopsep=0pt,parsep=0pt}
\usepackage{booktabs}
\usepackage{colortbl}
\usepackage{graphicx,overpic}

\tikzstyle{porte} = [fill=gray!25, draw]

\newtheorem{thm}{Theorem}
\newtheorem{defn}[thm]{Definition}
\newtheorem{lem}[thm]{Lemma}
\newtheorem{cor}[thm]{Corollary}
\newtheorem{question}{Question}

\theoremstyle{remark}
\newtheorem{rem}{Remark}

\numberwithin{equation}{section}

\DeclareMathOperator{\tr}{tr}
\DeclareMathOperator{\rank}{rank}

\DeclareMathOperator{\CPTP}{CPTP}

\newcommand{\ket}[1]{|#1\rangle}
\newcommand{\bra}[1]{\langle#1|}
\newcommand{\ketbra}[2]{\ket{#1}\!\bra{#2}}
\newcommand{\proj}[1]{\ketbra{#1}{#1}}

\allowdisplaybreaks

\begin{document}
\title{Entanglement fidelity of Petz decoder for one-shot entanglement transmission}
\author{Laura Burri}
\affiliation{Institute for Theoretical Physics, ETH Zurich, Zurich, Switzerland}

\begin{abstract}
One-shot entanglement transmission is a quantum information processing task where a quantum state is sent to a second party over a noisy channel. The goal of the task is to approximately recover the original state by applying a decoder to the output of the noisy channel. In this work, we note that the Petz map induces a universal decoder for one-shot entanglement transmission, and we quantify its entanglement fidelity. This fidelity is found to be determined by the singly minimized Petz R\'enyi mutual information of order $1/2$ associated with a complementary channel, thus providing an operational interpretation of this information measure. Furthermore, we compare the performance of this decoder to that of the decoder induced by the twirled Petz map and the Schumacher-Westmoreland decoder.
\end{abstract}

\maketitle


\section{Introduction}

Faithful transmission of quantum information between various parties is essential for quantum computation and communication. 
The analysis of related information processing tasks typically requires the use of appropriate information measures, which, in turn, gain operational significance through their relevance to these tasks. 
This paper focuses on the task of one-shot entanglement transmission and investigates which information measure characterizes this task when the Petz map is employed for achieving entanglement transmission.

\emph{One-shot entanglement transmission} is concerned with a scenario where a quantum state $\rho_A$ is sent to a second party over a noisy channel. 
The noisy channel is described by a completely positive, trace-preserving (CPTP) map from $A$ to $B$, denoted by $\mathcal{N}_{A\rightarrow B}$. 
The goal of the task is to subsequently apply a CPTP map from $B$ to $A$ in such a way that the final state on $A$ is approximately equal to the initial state $\rho_A$, as measured by the entanglement fidelity. 
This choice of performance measure is common in the literature on one-shot entanglement transmission~\cite{schumacher1996sending,schumacher1996quantum,schumacher2001approximate,fletcher2007optimum,reimpell2005iterative,reimpell2006commentoptimumquantumerror,datta2013one}, 
though it is, of course, not the only possible choice~\cite[Section~IV.C]{schumacher1996sending}. 

Previous works have analyzed the task of one-shot entanglement transmission from various perspectives. 
In~\cite{schumacher2001approximate}, a universal recovery map, henceforth referred to as the~\emph{Schumacher-Westmoreland (SW) decoder}, was constructed. 
In this context, ``universal'' means that the decoder's construction recipe is applicable to any given $(\rho_A,\mathcal{N}_{A\rightarrow B})$ and that the decoder achieves perfect entanglement transmission whenever possible. 
The question of optimal recovery was addressed in~\cite{fletcher2007optimum}, which showed that the maximum achievable entanglement fidelity for any fixed $(\rho_A,\mathcal{N}_{A\rightarrow B})$ can be computed via a semidefinite program -- a widely studied type of convex optimization problem for which efficient solution algorithms are known. 
Similarly,~\cite{reimpell2005iterative,reimpell2006commentoptimumquantumerror} examined the case where not only the recovery map but also the encoding is optimized. 
\cite{datta2013one,beigi2016decoding} investigated one-shot capacities for entanglement transmission.

Since the goal of one-shot entanglement transmission is to recover the information originally in $A$ from $B$, it is natural to view it as a specific instance of information recovery. 
It therefore seems pertinent to ask whether the Petz map~\cite{petz1986sufficient,petz1988sufficiency,ohya1993quantum,petz2003monotonicity}, a well-established general-purpose tool for information recovery, can be applied to this task and, if so, how well it performs. 
This work addresses these questions in the case where the performance measure is chosen to be the entanglement fidelity.

Specifically, we define a decoder for one-shot entanglement transmission based on the Petz map and quantify its performance in terms of the entanglement fidelity. 
Similarly, we define and analyze a decoder based on the twirled Petz map, which is a modification of the Petz map that achieves approximate recovery~\cite{junge2018universal}. 
The main result of this paper, stated as Theorem~\ref{thm:petz}, 
quantifies the performance of the corresponding decoders: the \emph{Petz decoder} and the \emph{twirled Petz decoder}. 
As a corollary, the theorem implies that the entanglement fidelity of the Petz decoder is determined by the singly minimized Petz R\'enyi mutual information of order $1/2$ associated with a complementary channel to $\mathcal{N}_{A\rightarrow B}$, 
and that the twirled Petz decoder never outperforms the Petz decoder (see Corollary~\ref{cor:petz}).

The analytical results presented in Theorem~\ref{thm:petz} and Corollary~\ref{cor:petz} do not provide any insight into how the performance of the Petz decoder compares to that of the SW decoder. 
This question is answered numerically for three specific settings: the 3-qubit bit-flip code for the bit-flip channel, the Leung-Nielsen-Chuang-Yamamoto (LNCY) 4-qubit code~\cite{leung1997approximate} for the amplitude damping channel, and the 5-qubit code~\cite{laflamme1996perfect} for the amplitude damping channel. 
These three settings were chosen for their simplicity and to align with similar existing literature~\cite{fletcher2007optimum,reimpell2006commentoptimumquantumerror}. 
For each setting, the performance of the Petz decoder, the twirled Petz decoder, the SW decoder, and an optimal decoder is evaluated, enabling a direct comparison (see Section~\ref{sec:numerics}).

\paragraph*{Outline.} 
The remainder of this paper is structured as follows. 
Section~\ref{sec:preliminaries} provides an overview of the notation employed in this work (\ref{ssec:notation}) and contains definitions of information measures (\ref{ssec:entropies}), 
the fidelity and entanglement fidelity (\ref{ssec:fidelity}), 
and the Petz map and twirled Petz map (\ref{ssec:petz}). 
Section~\ref{sec:main} addresses the topic of one-shot entanglement transmission and constitutes the core of this paper. 
First, previous results on the performance of the SW decoder are reviewed and improved (\ref{sec:sw-decoder}). 
Then, the Petz decoder and the twirled Petz decoder are defined, and their performance is analyzed (\ref{sec:petz-decoder}), leading to the main results of this paper (Theorem~\ref{thm:petz}, Corollary~\ref{cor:petz}). 
Finally, numerical results on the performance of various decoders are presented, enabling a direct comparison (\ref{sec:numerics}). 
The paper concludes with a discussion of the results in Section~\ref{sec:discussion}.

\section{Preliminaries}\label{sec:preliminaries}
\subsection{Notation}\label{ssec:notation}
``$\log$'' refers to the logarithm with base 2, and ``$\ln$'' refers to the natural logarithm. 
For any natural number $n$, the set of natural numbers strictly less than $n$ is denoted by $[n]\coloneqq\{0,1,\dots, n-1\}$.

Throughout this paper, all Hilbert spaces are assumed to be finite-dimensional for simplicity. 
The dimension of a Hilbert space $A$ will be denoted by $d_A$. 
The tensor product of two Hilbert spaces $A$ and $B$ is occasionally denoted by $AB$ instead of $A\otimes B$. 
The set of linear maps from $A$ to $B$ is denoted by $\mathcal{L}(A,B)$, and we define $\mathcal{L}(A)\coloneqq \mathcal{L}(A,A)$. 
To simplify the notation, identities are occasionally omitted, i.e., for any $X_A\in \mathcal{L}(A)$, ``$X_A$'' may be interpreted as $X_A\otimes 1_B\in \mathcal{L}(AB)$ where $1_B\in \mathcal{L}(B)$ denotes the identity operator on $B$. 
For $X,Y\in \mathcal{L}(A)$, $X\ll Y$ is true iff the kernel of $Y$ is contained in the kernel of $X$. 
For $X,Y\in \mathcal{L}(A)$, $X\perp Y$ is true iff $XY=YX=0$. 
The rank of $X\in \mathcal{L}(A)$ is denoted as $\rank(X)$.

For any $X\in \mathcal{L}(A,B)$, the adjoint of $X$ with respect to the inner products of $A$ and $B$ is denoted by $X^\dagger$. 
For $X\in \mathcal{L}(A)$, $X\geq 0$ is true iff $X$ is positive semidefinite. 
For a positive semidefinite $X\in \mathcal{L}(A)$, $X^p$ for $p\in \mathbb{R}$ is defined by taking the power on the support of $X$. 
The operator absolute value of $X\in \mathcal{L}(A)$ is $\lvert X\rvert \coloneqq (X^\dagger X)^{1/2}$. 
The trace of $X\in \mathcal{L}(A)$ is denoted as $\tr[X]$. 
The Schatten $p$-norm of $X\in \mathcal{L}(A)$ is defined as 
$\|X \|_p\coloneqq (\tr[\lvert X\rvert^p])^{1/p}$ for $p\in [1,\infty)$, 
and the Schatten $p$-quasi-norm is defined as $\|X \|_p\coloneqq (\tr[\lvert X\rvert^p])^{1/p}$ for $p\in (0,1)$.

The set of quantum states on $A$ is denoted by 
$\mathcal{S}(A)\coloneqq\{\rho\in \mathcal{L}(A):\rho\geq 0,\tr[\rho]=1\}$. 
The set of completely positive, trace-preserving linear maps from $\mathcal{L}(A)$ to $\mathcal{L}(B)$ is denoted by $\CPTP(A,B)$. 
Elements of this set are called (quantum) channels. 
The identity channel is denoted as $\mathcal{I}_{A\rightarrow A}$, and it is defined by $\mathcal{I}_{A\rightarrow A}(X_A)\coloneqq X_A$ for all $X_A\in \mathcal{L}(A)$. 
To simplify the notation, identity channels are occasionally omitted, i.e., for any $\mathcal{M}_{A\rightarrow B}\in \CPTP(A,B)$, ``$\mathcal{M}_{A\rightarrow B}$'' may be interpreted as $\mathcal{M}_{A\rightarrow B}\otimes \mathcal{I}_{E\rightarrow E}\in \CPTP(AE,BE)$. 
The partial trace over $A$ is denoted by $\tr_A$. 
For any fixed $\mathcal{M}_{A\rightarrow B}\in \CPTP(A,B)$, a map $\mathcal{M}_{A\rightarrow E}^c\in \CPTP(A,E)$ is said to be a complementary channel to $\mathcal{M}_{A\rightarrow B}$ if there exists an isometry $V\in \mathcal{L}(A,BE)$ such that 
$\mathcal{M}_{A\rightarrow B}(X_A)= \tr_E[VX_AV^\dagger] $ and 
$\mathcal{M}_{A\rightarrow E}^c(X_A)= \tr_B[VX_AV^\dagger] $ for all $X_A\in \mathcal{L}(A)$.

\subsection{Entropies, divergences, R\'enyi mutual information}\label{ssec:entropies}
The \emph{von Neumann entropy} of $\rho\in \mathcal{S}(A)$ is $H(A)_{\rho}\coloneqq -\tr[\rho\log \rho]$, 
and the \emph{R\'enyi entropy (of order $\alpha$)} is $H_{\alpha}(A)_\rho\coloneqq \frac{1}{1-\alpha}\log\tr[\rho^{\alpha}]$ for $\alpha \in (0,1)\cup (1,\infty)$. 
For a bipartite quantum state $\rho_{AB}\in \mathcal{S}(AB)$, 
the \emph{conditional entropy} is $H(A|B)_\rho\coloneqq H(AB)_\rho - H(B)_\rho$, 
the \emph{mutual information} is $I(A:B)_\rho\coloneqq H(A)_\rho +H(B)_\rho - H(AB)_\rho$, 
and the \emph{coherent information} is $I(A\, \rangle B)_{\rho}\coloneqq -H(A|B)_\rho $.

The \emph{(quantum) relative entropy} of $\rho\in \mathcal{S}(A)$ relative to a positive semidefinite $\sigma\in \mathcal{L}(A)$ is 
\begin{align}
D(\rho\| \sigma) \coloneqq \tr[\rho (\log \rho - \log \sigma)]
\end{align}
if $\rho\ll \sigma$, 
and $D(\rho\| \sigma) \coloneqq \infty$ else.

The \emph{Petz (quantum R\'enyi) divergence (of order $\alpha$)} of $\rho\in \mathcal{S}(A)$ relative to a positive semidefinite $\sigma\in \mathcal{L}(A)$ is defined for $\alpha\in (0,1)\cup (1,\infty)$ as~\cite{petz1986quasi}
\begin{equation}
D_\alpha (\rho\| \sigma)\coloneqq\frac{1}{\alpha -1} \log \tr [\rho^\alpha \sigma^{1-\alpha}]
\end{equation}
if $(\alpha <1\land \rho\not\perp\sigma)\lor \rho\ll \sigma$, 
and $D_\alpha (\rho\| \sigma)\coloneqq \infty$ else. 
Moreover, $D_\alpha$ is defined for $\alpha\in \{0,1\}$ as the respective limit of $D_\alpha$ for $\alpha\rightarrow\{0,1\}$. 
The limit $\alpha\rightarrow 1$ reduces to the quantum relative entropy, i.e., 
$D_1 (\rho\| \sigma)=D(\rho\| \sigma)$~\cite{lin2015investigating,tomamichel2016quantum}.

The \textit{sandwiched (quantum R\'enyi) divergence (of order $\alpha$)} of $\rho\in \mathcal{S}(A)$ relative to a positive semidefinite $\sigma\in \mathcal{L}(A)$ is defined for $\alpha\in (0,1)\cup (1,\infty)$ as~\cite{mueller2013quantum,wilde2014strong}
\begin{align}
\widetilde{D}_\alpha (\rho\| \sigma)&\coloneqq 
\frac{1}{\alpha -1}\log 
\lVert\sigma^{\frac{1-\alpha}{2\alpha}} \rho  \sigma^{\frac{1-\alpha}{2\alpha}}\rVert^\alpha_\alpha
\end{align}
if $(\alpha <1\land \rho\not\perp\sigma)\lor \rho\ll \sigma$, 
and $\widetilde{D}_\alpha (\rho\| \sigma)\coloneqq \infty$ else. 
Moreover, $\widetilde{D}_\alpha$ is defined for $\alpha\in \{1,\infty\}$ as the respective limit of $\widetilde{D}_\alpha$ for $\alpha\rightarrow\{1,\infty\}$. 
The limit $\alpha\rightarrow 1$ reduces to the quantum relative entropy, i.e., 
$\widetilde{D}_1 (\rho\| \sigma)=D(\rho\| \sigma)$~\cite{mueller2013quantum,wilde2014strong,tomamichel2016quantum}.

The \emph{minimized generalized Petz R\'enyi mutual information (of order $\alpha$)} of $\rho_{AB}\in \mathcal{S}(AB)$ relative to a positive semidefinite $\sigma_A\in \mathcal{L}(A)$ is defined for $\alpha\in [0,\infty)$ as~\cite{hayashi2016correlation,burri2024doublyminimizedpetzrenyi}
\begin{align}
I_\alpha^\downarrow(\rho_{AB}\| \sigma_A)\coloneqq \inf_{\tau_B\in \mathcal{S}(B)}D_\alpha (\rho_{AB}\| \sigma_A\otimes \tau_B ) .
\end{align}

The \emph{singly minimized Petz R\'enyi mutual information (of order $\alpha$)} of $\rho_{AB}\in \mathcal{S}(AB)$ is defined for $\alpha\in [0,\infty)$ as~\cite{gupta2014multiplicativity,hayashi2016correlation,burri2024doublyminimizedpetzrenyi}
\begin{align}
I_\alpha^{\uparrow\downarrow}(A:B)_\rho 
\coloneqq \inf_{\tau_B\in \mathcal{S}(B)} D_\alpha (\rho_{AB}\| \rho_A\otimes \tau_B)
=I_\alpha^\downarrow (\rho_{AB}\| \rho_A) .
\end{align}

The \emph{non-minimized generalized sandwiched R\'enyi mutual information (of order $\alpha$)} of $\rho_{AB}\in \mathcal{S}(AB)$ relative to a positive semidefinite $\sigma_A\in \mathcal{L}(A)$ is defined for $\alpha\in (0,\infty]$ as~\cite{burri2024doublyminimizedsandwichedrenyi}
\begin{align}
\widetilde{I}_\alpha^\uparrow(\rho_{AB}\| \sigma_A)\coloneqq \widetilde{D}_\alpha (\rho_{AB}\| \sigma_A\otimes \rho_B ) .
\end{align}

The \emph{non-minimized sandwiched R\'enyi mutual information (of order $\alpha$)} of $\rho_{AB}\in \mathcal{S}(AB)$ is defined for $\alpha\in (0,\infty]$ as~\cite{burri2024doublyminimizedsandwichedrenyi}
\begin{align}
\widetilde{I}_\alpha^{\uparrow\uparrow}(A:B)_\rho \coloneqq \widetilde{D}_\alpha (\rho_{AB}\| \rho_A\otimes \rho_B)
=\widetilde{I}_\alpha^\uparrow(\rho_{AB}\| \rho_A).
\end{align}

\subsection{Fidelity and entanglement fidelity}\label{ssec:fidelity}
The \emph{fidelity} between $\rho,\sigma\in \mathcal{S}(A)$ is 
$F(\rho,\sigma)\coloneqq \|\rho^{\frac{1}{2}} \sigma^{\frac{1}{2}} \|_1$.

Let $\rho_A\in \mathcal{S}(A)$ and let $\mathcal{M}_{A\rightarrow A}\in \CPTP(A,A)$. 
Let $\ket{\rho}_{RA}\in RA$ be such that $\tr_R[\proj{\rho}_{RA}]=\rho_A$. 
Then the \emph{entanglement fidelity} of $\mathcal{M}_{A\rightarrow A}$ with respect to $\rho_A$ is~\cite{schumacher1996sending}
\begin{align}\label{eq:def-ef}
F_e(\rho_A,\mathcal{M}_{A\rightarrow A})
&\coloneqq F^2(\proj{\rho}_{RA},\mathcal{M}_{A\rightarrow A}(\proj{\rho}_{RA}))
\\
&=\bra{\rho}_{RA} \mathcal{M}_{A\rightarrow A}(\proj{\rho}_{RA}) \ket{\rho}_{RA} .
\end{align}
Note that the entanglement fidelity is well-defined, as the right-hand side of~\eqref{eq:def-ef} remains invariant under different choices of the purification $\ket{\rho}_{RA}$ of $\rho_A$~\cite{schumacher1996sending}. 
Thus, the left-hand side of~\eqref{eq:def-ef} only depends on $\rho_A$ and $\mathcal{M}_{A\rightarrow A}$.

\subsection{Petz map and twirled Petz map}\label{ssec:petz}
The \emph{data processing inequality (for the relative entropy)}~\cite{lindblad1975completely,uhlmann1977relative} asserts that for any $\mathcal{N}\in \CPTP(A,B),\rho,\sigma\in \mathcal{S}(A)$
\begin{align}\label{eq:monotonicity}
D(\rho\| \sigma)\geq D(\mathcal{N}(\rho)\| \mathcal{N}(\sigma)).
\end{align}

Let $\mathcal{N}\in \CPTP(A,B)$ and let $\rho,\sigma\in \mathcal{S}(A)$ be such that $\rho\ll \sigma$. 
As shown in~\cite{petz1986sufficient,petz1988sufficiency,ohya1993quantum,petz2003monotonicity}, if the data processing inequality is saturated, i.e., 
$D(\rho\| \sigma)=D(\mathcal{N}(\rho)\| \mathcal{N}(\sigma))$, 
then there exists $\mathcal{R}\in \CPTP(B,A)$ such that 
$\mathcal{R}\circ \mathcal{N}(\sigma)=\sigma$ and $\mathcal{R}\circ \mathcal{N}(\rho)=\rho$. 
Moreover, an explicit and $\rho$-independent form for the recovery map is known~\cite{petz1986sufficient,petz1988sufficiency,ohya1993quantum,petz2003monotonicity}: 
The aforementioned properties are satisfied if the recovery map $\mathcal{R}$ is given by the \emph{Petz (recovery) map}, which is defined on the support of $\mathcal{N}(\sigma)$ as
\begin{align}\label{eq:def_P}
\mathcal{P}_{\sigma,\mathcal{N}}:\quad  X_B\mapsto 
\sigma^{\frac{1}{2}}\mathcal{N}^\dagger (\mathcal{N}(\sigma)^{-\frac{1}{2}} X_B \mathcal{N}(\sigma)^{-\frac{1}{2}}) \sigma^{\frac{1}{2}}.
\end{align}

Let $\mathcal{N}\in \CPTP(A,B)$ and let $\rho,\sigma\in \mathcal{S}(A)$ be such that $\rho\ll \sigma$. 
As shown in~\cite{wilde2015recoverability,sutter2016strengthened}, this implies that there exists $\mathcal{R}\in \CPTP(B,A)$ such that $\mathcal{R}\circ \mathcal{N}(\sigma)=\sigma$ and
\begin{align}\label{eq:d-recovery}
D(\rho\| \sigma)-D(\mathcal{N}(\rho)\| \mathcal{N}(\sigma))\geq -2\log F(\rho, \mathcal{R}\circ \mathcal{N}(\rho)).
\end{align}
This inequality is a strengthened version of~\eqref{eq:monotonicity} 
since the right-hand side of~\eqref{eq:d-recovery} is non-negative. 
Moreover, an explicit and $\rho$-independent form for the recovery map was found~\cite{junge2018universal}, which we refer to as the \textit{twirled Petz (recovery) map}. 
It is defined on the support of $\mathcal{N}(\sigma)$ as
\begin{align}\label{eq:def_R}
\mathcal{R}_{\sigma,\mathcal{N}}:\quad X_B\mapsto 
\int_{-\infty}^{\infty}\mathrm{d}t\, \beta_0(t) \mathcal{R}_{\sigma,\mathcal{N}}^{t/2}(X_B)
\end{align}
where
\begin{align}\label{eq:def_Rt}
\mathcal{R}_{\sigma,\mathcal{N}}^t:\quad X_B\mapsto 
\sigma^{\frac{1}{2}-it}\mathcal{N}^\dagger(\mathcal{N}(\sigma)^{-\frac{1}{2}+it} X_B \mathcal{N}(\sigma)^{-\frac{1}{2}-it})\sigma^{\frac{1}{2}+it}
\end{align}
and $\beta_0$ is the probability density function defined by
\begin{align}
\beta_0(t)\coloneqq \frac{\pi}{2}(\cosh(\pi t)+1)^{-1} 
\qquad\forall t\in \mathbb{R} .
\label{eq:def-beta0}
\end{align}

\section{One-shot entanglement transmission}\label{sec:main}

\paragraph*{Problem definition.}  
The task of \emph{one-shot entanglement transmission} for any given $(\rho_A,\mathcal{N}_{A\rightarrow B})\in \mathcal{S}(A)\times \CPTP(A,B)$ consists in the construction of a \emph{decoder (or: recovery map)} $\mathcal{D}_{B\rightarrow A}\in \CPTP(B,A)$. 
The performance of the decoder is measured by the entanglement fidelity 
\begin{align}\label{eq:def_fe}
F_e(\rho_A,\mathcal{D}_{B\rightarrow A}\circ\mathcal{N}_{A\rightarrow B}).
\end{align}
The greater the entanglement fidelity, the better the performance of the decoder. 
The task of one-shot entanglement transmission is depicted in Figure~\ref{fig:scheme-entanglement-transmission}.

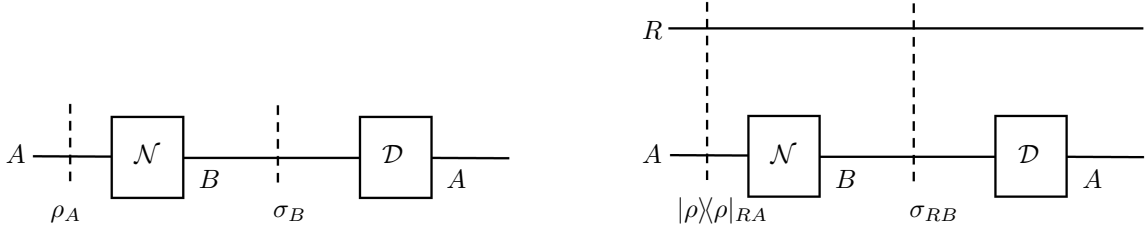
\begin{figure}\centering
\begin{overpic}[width=.4\textwidth]{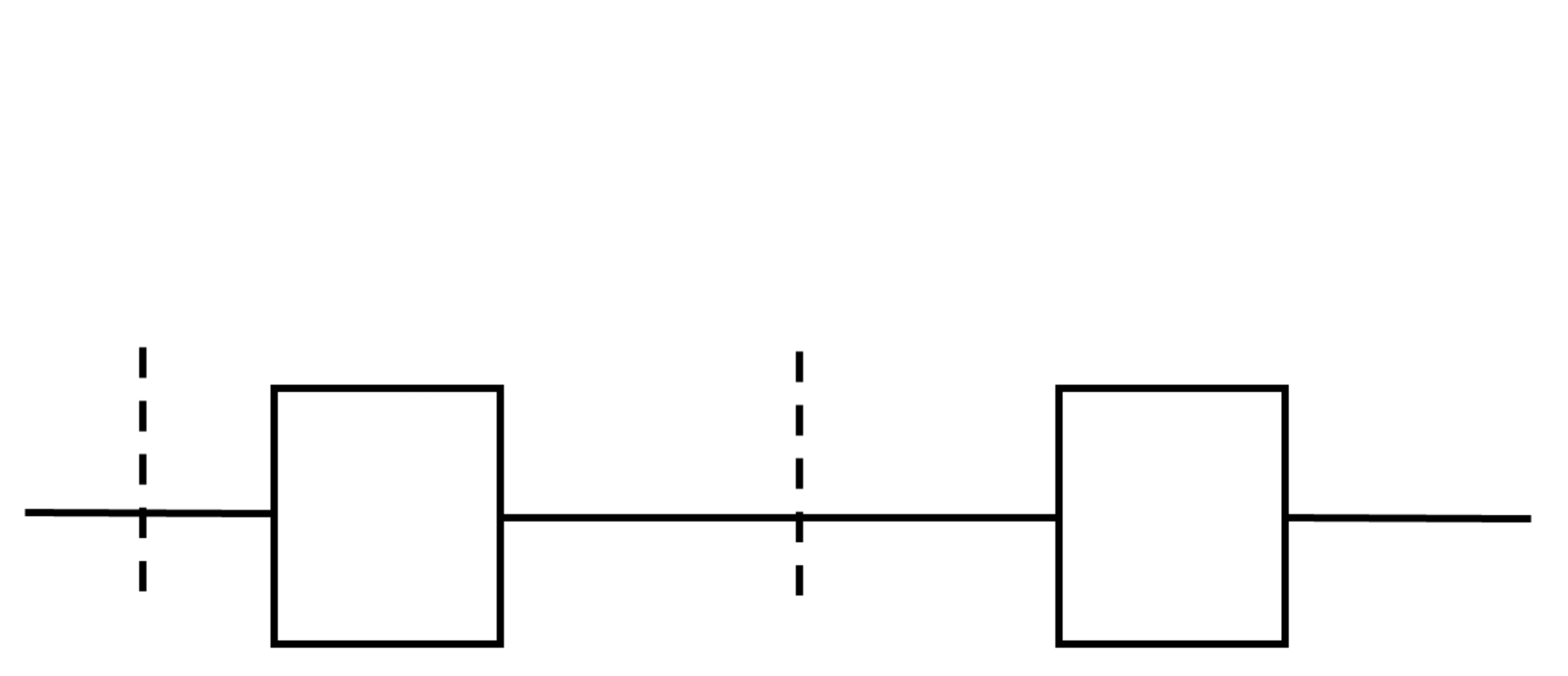}
  \put(-4,9){$A$}
  \put(5,-2){$\rho_A$}
  \put(22,9){$\mathcal{N}$}
  \put(35,4){$B$}
  \put(50,-2){$\sigma_{B}$}
  \put(72,9){$\mathcal{D}$}
  \put(85,4){$A$}
\end{overpic}
\hspace*{.1\textwidth}
\begin{overpic}[width=.4\textwidth]{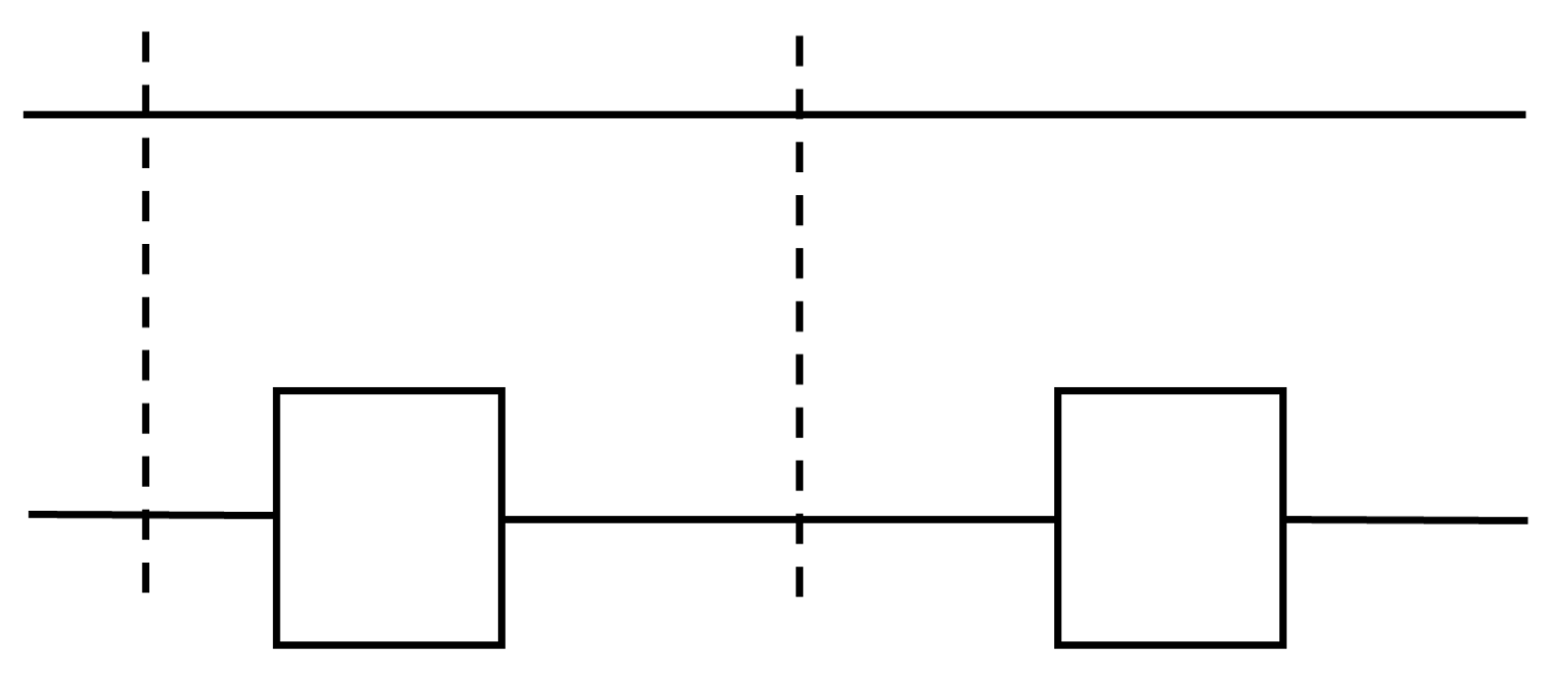}
  \put(-4,34){$R$}
  \put(-4,9){$A$}
  \put(3,-2){$\proj{\rho}_{RA}$}
  \put(22,9){$\mathcal{N}$}
  \put(35,4){$B$}
  \put(50,-2){$\sigma_{RB}$}
  \put(72,9){$\mathcal{D}$}
  \put(85,4){$A$}
\end{overpic}
\vspace*{1em}
\caption{
\textbf{Left.} Coding scheme for one-shot entanglement transmission of a quantum state $\rho_A$ over the channel $\mathcal{N}_{A\rightarrow B}$. 
The intermediate state is given by $\sigma_B\coloneqq \mathcal{N}_{A\rightarrow B}(\rho_A)$. 
According to~\eqref{eq:def_fe}, the performance of the decoder $\mathcal{D}_{B\rightarrow A}$ is measured by the entanglement fidelity $F_e(\rho_A,\mathcal{D}_{B\rightarrow A}\circ \mathcal{N}_{A\rightarrow B})$. 
\textbf{Right.} 
Let $\ket{\rho}_{RA}\in RA$ be such that $\tr_R[\proj{\rho}_{RA}]=\rho_A$. 
The initial state is then $\proj{\rho}_{RA}$ and the intermediate state is $\sigma_{RB}\coloneqq \mathcal{N}_{A\rightarrow B}(\proj{\rho}_{RA})$. 
The squared fidelity between the initial and final state is 
$F^2(\proj{\rho}_{RA},\mathcal{D}_{B\rightarrow A}\circ\mathcal{N}_{A\rightarrow B}(\proj{\rho}_{RA}))$. 
According to~\eqref{eq:def-ef}, this squared fidelity (illustrated on the right-hand side) is identical to the entanglement fidelity $F_e(\rho_A,\mathcal{D}_{B\rightarrow A}\circ \mathcal{N}_{A\rightarrow B})$ (illustrated on the left-hand side).
}
\label{fig:scheme-entanglement-transmission}
\end{figure}

\begin{rem}[Optimal decoder] 
For any fixed $(\rho_A,\mathcal{N}_{A\rightarrow B})\in \mathcal{S}(A)\times \CPTP(A,B)$, the \emph{maximal achievable entanglement fidelity} is 
\begin{align}\label{eq:optimal}
\max_{\mathcal{D}_{B\rightarrow A}\in \CPTP(B,A)}
F_e(\rho_A,\mathcal{D}_{B\rightarrow A}\circ\mathcal{N}_{A\rightarrow B}).
\end{align}
If $\mathcal{D}_{B\rightarrow A}\in \CPTP(B,A)$ is an optimizer for this optimization problem, then $\mathcal{D}_{B\rightarrow A}$ is said to be \emph{optimal}. 
Remarkably, the optimization problem~\eqref{eq:optimal} can be expressed as a semidefinite program~\cite{fletcher2007optimum}. 
Since efficient algorithms exist for solving semidefinite programs, this implies that~\eqref{eq:optimal} can often be computed efficiently. 
In particular, this will be the case for the settings studied numerically in Section~\ref{sec:numerics}.
\end{rem}

\subsection{Schumacher-Westmoreland decoder}\label{sec:sw-decoder}

In Appendix~\ref{app:sw-construction}, we review the construction of the SW decoder~\cite{schumacher2001approximate} for any given $(\rho_A,\mathcal{N}_{A\rightarrow B})\in \mathcal{S}(A)\times \CPTP(A,B)$. 
Building on this, this section presents new results on the performance of the SW decoder. 
To assess the performance of the SW decoder and facilitate a direct comparison with other decoders, it is desirable to express its entanglement fidelity -- or at least a lower bound of it -- as a function of $\sigma_{RB}\coloneqq\mathcal{N}_{A\rightarrow B}(\proj{\rho}_{RA})$. 
The following theorem establishes such a bound. 
The proof of Theorem~\ref{thm:sw} is given in Appendix~\ref{app:sw-proof}.

\begin{thm}[Entanglement fidelity of SW decoder]\label{thm:sw}
Let $(\rho_A,\mathcal{N}_{A\rightarrow B})\in \mathcal{S}(A)\times \CPTP(A,B)$. 
Let $\mathcal{D}^{\mathrm{SW}}_{B\rightarrow A}$ be the SW decoder for $(\rho_A,\mathcal{N}_{A\rightarrow B})$. 
Let $\ket{\rho}_{RA}\in RA$ be such that $\tr_R[\proj{\rho}_{RA}]=\rho_A$ and let 
$\sigma_{RB}\coloneqq \mathcal{N}_{A\rightarrow B}(\proj{\rho}_{RA})$. 
Then
\begin{align}\label{eq:fidelity-sw}
\log F_e(\rho_A,\mathcal{D}^{\mathrm{SW}}_{B\rightarrow A}\circ\mathcal{N}_{A\rightarrow B})
\geq I_2^\downarrow(\sigma_{RB}\| \sigma_R^{-1}).
\end{align}
\end{thm}

The following corollary follows from Theorem~\ref{thm:sw}, as proven in Appendix~\ref{app:sw-proof-cor}.
\begin{cor}[Entanglement fidelity of SW decoder]\label{cor:sw}
Let $(\rho_A,\mathcal{N}_{A\rightarrow B})\in \mathcal{S}(A)\times \CPTP(A,B)$. 
Let $\mathcal{D}^{\mathrm{SW}}_{B\rightarrow A}$ be the SW decoder for $(\rho_A,\mathcal{N}_{A\rightarrow B})$. 
Let $\ket{\rho}_{RA}\in RA$ be such that $\tr_R[\proj{\rho}_{RA}]=\rho_A$ and let 
$\sigma_{RB}\coloneqq \mathcal{N}_{A\rightarrow B}(\proj{\rho}_{RA})$. 
Then all of the following hold.
\begin{enumerate}[label=(\alph*)]
\item \emph{Duality:} 
Let $\ket{\sigma}_{RBE}\in RBE$ be such that $\tr_E[\proj{\sigma}_{RBE}]=\sigma_{RB}$ 
and let $\mathcal{N}_{A\rightarrow E}^c$ be a complementary channel to $\mathcal{N}_{A\rightarrow B}$. Then,
\begin{align}
\log F_e(\rho_A,\mathcal{D}^{\mathrm{SW}}_{B\rightarrow A}\circ\mathcal{N}_{A\rightarrow B})
\geq I_2^\downarrow(\sigma_{RB}\| \sigma_R^{-1})
&= -\widetilde{I}_{1/2}^{\uparrow\uparrow}(R:E)_\sigma
\label{eq:duality-sw}\\
&=-\widetilde{I}_{1/2}^{\uparrow\uparrow}(R:E)_{\mathcal{N}_{A\rightarrow E}^c(\proj{\rho}_{RA})}.
\label{eq:duality-sw2}
\end{align}
\item \emph{Weaker bound:} Let 
\begin{align}\label{eq:def-epsilon}
\varepsilon^\mathrm{SW}
\coloneqq -D(\sigma_{RB}\| \sigma_R^{-1}\otimes \sigma_B)
&=H(R)_{\sigma}-I(R\,\rangle\, B)_{\sigma}
=H(R|B)_{\sigma}-H(R|A)_{\rho}.
\end{align}
Then, 
\begin{align}\label{eq:sw-original}
\log F_e(\rho_A,\mathcal{D}^{\mathrm{SW}}_{B\rightarrow A}\circ\mathcal{N}_{A\rightarrow B})
\geq I_2^\downarrow(\sigma_{RB}\| \sigma_R^{-1})
&\geq -\varepsilon^{\mathrm{SW}}.
\end{align}
\end{enumerate}
\end{cor}

\begin{rem}[Comparison with original SW lower bound]\label{rem:original}
Our bounds in Theorem~\ref{thm:sw} and Corollary~\ref{cor:sw} may be compared to the bounds derived in the original work~\cite{schumacher2001approximate} on the SW decoder. 
The original work contains a proof of the lower bound 
$F_e^{1/2}(\rho_A,\mathcal{D}_{B\rightarrow A}^{\mathrm{SW}}\circ\mathcal{N}_{A\rightarrow B}) \geq 1-(\frac{\ln (2)}{2}\varepsilon^\mathrm{SW})^{1/2}$. 
Since $2^{-\varepsilon/2} \geq 1-(\frac{\ln (2)}{2}\varepsilon)^{1/2}$ for all $\varepsilon\in [0,\infty)$, it follows that our bounds in~\eqref{eq:fidelity-sw} and~\eqref{eq:sw-original} are improved versions of the original lower bound. 
We note that our proof of~\eqref{eq:fidelity-sw} follows the same proof technique as~\cite{schumacher2001approximate}. 
The only significant qualitative difference between our proof and that of~\cite{schumacher2001approximate} is our use of a different duality relation, which was not known when~\cite{schumacher2001approximate} was written. 
More details can be found in Appendix~\ref{app:original}.
\end{rem}

\subsection{Petz decoder and twirled Petz decoder}\label{sec:petz-decoder}

This section examines the usefulness of the Petz map and the twirled Petz map for one-shot entanglement transmission. 
We begin by defining the two corresponding decoders.

\begin{defn}[Definition of Petz and twirled Petz decoder]
Let $(\rho_A,\mathcal{N}_{A\rightarrow B})\in \mathcal{S}(A)\times \CPTP(A,B)$. 
The \emph{Petz decoder for $(\rho_A,\mathcal{N}_{A\rightarrow B})$} is 
$\mathcal{D}^\mathrm{P}_{B\rightarrow A}
\coloneqq \mathcal{P}_{\rho_A,\mathcal{N}_{A\rightarrow B}}$, see~\eqref{eq:def_P}.
The \emph{twirled Petz decoder for $(\rho_A,\mathcal{N}_{A\rightarrow B})$} is 
$\mathcal{D}^\mathrm{P,twirled}_{B\rightarrow A}
\coloneqq \mathcal{R}_{\rho_A,\mathcal{N}_{A\rightarrow B}}
=\int_{-\infty}^{\infty} \mathrm{d}t\, \beta_0(t) \mathcal{R}_{\rho_A,\mathcal{N}_{A\rightarrow B}}^{t/2}$, see \eqref{eq:def_R}--\eqref{eq:def-beta0}.
\end{defn}

To assess the performance of the Petz and the twirled Petz decoder, it is useful to express their entanglement fidelity as a function of $\sigma_{RB}\coloneqq\mathcal{N}_{A\rightarrow B}(\proj{\rho}_{RA})$. 
This is achieved by the following theorem. 
The proof of Theorem~\ref{thm:petz} is given in Appendix~\ref{app:proof-thm-petz}.

\begin{thm}[Entanglement fidelity of Petz and twirled Petz decoder]\label{thm:petz}
Let $(\rho_A,\mathcal{N}_{A\rightarrow B})\in \mathcal{S}(A)\times \CPTP(A,B)$. 
Let $\mathcal{D}^\mathrm{P}_{B\rightarrow A}$ and $\mathcal{D}^{\mathrm{P,twirled}}_{B\rightarrow A}$ be the Petz and the twirled Petz decoder for $(\rho_A,\mathcal{N}_{A\rightarrow B})$, 
and let $\mathcal{D}^{\mathrm{P},t}_{B\rightarrow A}\coloneqq \mathcal{R}_{\rho_A,\mathcal{N}_{A\rightarrow B}}^{t/2}$. 
Let $\ket{\rho}_{RA}\in RA$ be such that $\tr_R[\proj{\rho}_{RA}]=\rho_A$ 
and let $\sigma_{RB}\coloneqq \mathcal{N}_{A\rightarrow B}(\proj{\rho}_{RA})$. 
Then 
\begin{align}
\log F_e(\rho_A,\mathcal{D}^{\mathrm{P,}t}_{B\rightarrow A}\circ\mathcal{N}_{A\rightarrow B})
&=\log \lVert \sigma_{RB}^{\frac{1}{2}} (\sigma_{R}^{\frac{1}{2}(1+it)}\otimes \sigma_{B}^{-\frac{1}{2}(1+it)}) \sigma_{RB}^{\frac{1}{2}} \rVert_2^2,
\label{eq:thm-t}
\\
\log F_e(\rho_A,\mathcal{D}^{\mathrm{P}}_{B\rightarrow A}\circ\mathcal{N}_{A\rightarrow B})
&=\log \lVert \sigma_{RB}^{\frac{1}{2}} (\sigma_{R}^{\frac{1}{2}}\otimes \sigma_{B}^{-\frac{1}{2}}) \sigma_{RB}^{\frac{1}{2}} \rVert_2^2
=\widetilde{I}_2^\uparrow(\sigma_{RB}\| \sigma_R^{-1}),
\label{eq:thm-petz}\\
\log F_e(\rho_A,\mathcal{D}^{\mathrm{P,twirled}}_{B\rightarrow A}\circ\mathcal{N}_{A\rightarrow B})
&=\log \int_{-\infty}^{\infty}\mathrm{d}t\, \beta_0(t) \lVert \sigma_{RB}^{\frac{1}{2}} (\sigma_{R}^{\frac{1}{2}(1+it)}\otimes \sigma_{B}^{-\frac{1}{2}(1+it)}) \sigma_{RB}^{\frac{1}{2}} \rVert_2^2.
\label{eq:thm-twirled}
\end{align}
\end{thm}

The following corollary follows from Theorem~\ref{thm:petz}, as proven in Appendix~\ref{app:proof-cor-petz}.
\begin{cor}[Entanglement fidelity of Petz and twirled Petz decoder]\label{cor:petz}
Let $(\rho_A,\mathcal{N}_{A\rightarrow B})\in \mathcal{S}(A)\times \CPTP(A,B)$. 
Let $\mathcal{D}^\mathrm{P}_{B\rightarrow A}$ and $\mathcal{D}^{\mathrm{P,twirled}}_{B\rightarrow A}$ be the Petz and the twirled Petz decoder for $(\rho_A,\mathcal{N}_{A\rightarrow B})$, 
and let $\mathcal{D}^{\mathrm{P},t}_{B\rightarrow A}\coloneqq \mathcal{R}_{\rho_A,\mathcal{N}_{A\rightarrow B}}^{t/2}$. 
Let $\ket{\rho}_{RA}\in RA$ be such that $\tr_R[\proj{\rho}_{RA}]=\rho_A$ 
and let $\sigma_{RB}\coloneqq \mathcal{N}_{A\rightarrow B}(\proj{\rho}_{RA})$. 
Then all of the following hold.
\begin{enumerate}[label=(\alph*)]
\item \emph{Duality for Petz decoder:} 
Let $\ket{\sigma}_{RBE}\in RBE$ be such that $\tr_E[\proj{\sigma}_{RBE}]=\sigma_{RB}$ and 
let $\mathcal{N}_{A\rightarrow E}^c$ be a complementary channel to $\mathcal{N}_{A\rightarrow B}$. Then, 
\begin{align}
\log F_e(\rho_A,\mathcal{D}^{\mathrm{P}}_{B\rightarrow A}\circ\mathcal{N}_{A\rightarrow B})
=\widetilde{I}_2^\uparrow(\sigma_{RB}\| \sigma_R^{-1})
&=-I_{1/2}^{\uparrow\downarrow}(R:E)_\sigma
\label{eq:thm-duality0}\\
&=-I_{1/2}^{\uparrow\downarrow}(R:E)_{\mathcal{N}_{A\rightarrow E}^c(\proj{\rho}_{RA})}.
\label{eq:thm-duality}
\end{align}
\item \emph{Optimality of $t=0$:} 
$\mathcal{D}^{\mathrm{P}}_{B\rightarrow A}=\mathcal{D}^{\mathrm{P},0}_{B\rightarrow A}$ and for all $t\in \mathbb{R}$
\begin{align}
F_e(\rho_A,\mathcal{D}^{\mathrm{P},0}_{B\rightarrow A}\circ\mathcal{N}_{A\rightarrow B})
\geq 
F_e(\rho_A,\mathcal{D}^{\mathrm{P},t}_{B\rightarrow A}\circ\mathcal{N}_{A\rightarrow B}).
\end{align}
\item \emph{Petz decoder outperforms twirled Petz decoder:} 
\begin{align}
\log F_e(\rho_A,\mathcal{D}^{\mathrm{P}}_{B\rightarrow A}\circ\mathcal{N}_{A\rightarrow B})
&\geq \log F_e(\rho_A,\mathcal{D}^{\mathrm{P,twirled}}_{B\rightarrow A}\circ\mathcal{N}_{A\rightarrow B})
\label{eq:thm-bound1}\\
&\geq D(\sigma_{RB}\| \sigma_R^{-1}\otimes \sigma_B).
\label{eq:thm-bound3}
\end{align}
\end{enumerate}
\end{cor}

\begin{rem}[Comparison with SW lower bound]\label{rem:sw-petz}
The lower bound~\eqref{eq:thm-bound3} can be expressed as 
$\log F_e(\rho_A,\mathcal{D}^{\mathrm{P,twirled}}_{B\rightarrow A}\circ\mathcal{N}_{A\rightarrow B})\geq -\varepsilon^{\mathrm{SW}}$ where $\varepsilon^{\mathrm{SW}}\coloneqq -D(\sigma_{RB}\| \sigma_R^{-1}\otimes \sigma_B)$. 
This lower bound also applies to the SW decoder, see Corollary~\ref{cor:sw}~(b).
\end{rem}
\begin{rem}[Petz decoder is near-optimal] \label{rem:optimal}
For any fixed $(\rho_A,\mathcal{N}_{A\rightarrow B})\in \mathcal{S}(A)\times \CPTP(A,B)$, 
\begin{align}
\max_{\mathcal{D}_{B\rightarrow A}\in \CPTP(B,A)} F_e^2(\rho_A,\mathcal{D}_{B\rightarrow A}\circ \mathcal{N}_{A\rightarrow B})
\leq F_e(\rho_A,\mathcal{D}_{B\rightarrow A}^\mathrm{P}\circ \mathcal{N}_{A\rightarrow B})
\label{eq:barnum-knill}
\end{align}
where $\mathcal{D}_{B\rightarrow A}^\mathrm{P}$ denotes the Petz decoder for $(\rho_A,\mathcal{N}_{A\rightarrow B})$~\cite[Corollary~2]{barnum2002reversing} (see also~\cite[Theorem~3]{ng2010simple}). 
\eqref{eq:barnum-knill} implies that the Petz decoder performs nearly as well as an optimal decoder. 
Conversely,~\eqref{eq:barnum-knill} provides a useful upper bound on the performance of an optimal decoder, as its right-hand side admits a closed-form expression in terms of $\sigma_{RB}$, as shown in~\eqref{eq:thm-petz} in Theorem~\ref{thm:petz}.
\end{rem}

As mentioned in the introduction, the SW decoder achieves perfect recovery whenever possible~\cite{schumacher1996quantum,schumacher2001approximate}.  
Using Theorem~\ref{thm:petz}, one can similarly show that the Petz and the twirled Petz decoder also achieve perfect recovery whenever possible. 
The following corollary formalizes this claim and provides additional equivalence conditions for perfect recovery. 
A proof of Corollary~\ref{cor:perfect} is given in Appendix~\ref{app:corollary}. 
We note that the equivalence of~(b), (e), and~(f) follows from previous work~\cite{schumacher1996quantum,schumacher2001approximate}.

\begin{cor}[Perfect one-shot entanglement transmission]\label{cor:perfect}
Let $(\rho_A,\mathcal{N}_{A\rightarrow B})\in \mathcal{S}(A)\times \CPTP(A,B)$. 
Let $\ket{\rho}_{RA}\in RA$ be such that $\tr_R[\proj{\rho}_{RA}]=\rho_A$ 
and let $\sigma_{RB}\coloneqq \mathcal{N}_{A\rightarrow B}(\proj{\rho}_{RA})$. 
Then the following assertions are equivalent.
\begin{enumerate}[label=(\alph*)]
\item \emph{Achievability of perfect recovery:} \\
There exists $\mathcal{D}_{B\rightarrow A}\in \CPTP(B,A)$ such that 
$F_e(\rho_A,\mathcal{D}_{B\rightarrow A}\circ\mathcal{N}_{A\rightarrow B})
=1$.
\item \emph{SW decoder achieves perfect recovery:} \\
The SW decoder $\mathcal{D}_{B\rightarrow A}^{\mathrm{SW}}$ for $(\rho_A,\mathcal{N}_{A\rightarrow B})$ is such that 
$F_e(\rho_A,\mathcal{D}_{B\rightarrow A}^{\mathrm{SW}}\circ\mathcal{N}_{A\rightarrow B})
=1$.
\item \emph{Petz decoder achieves perfect recovery:} \\
The Petz decoder $\mathcal{D}_{B\rightarrow A}^{\mathrm{P}}$ for $(\rho_A,\mathcal{N}_{A\rightarrow B})$ is such that 
$F_e(\rho_A,\mathcal{D}_{B\rightarrow A}^{\mathrm{P}}\circ\mathcal{N}_{A\rightarrow B})
=1$.
\item \emph{Twirled Petz decoder achieves perfect recovery:} \\
The twirled Petz decoder $\mathcal{D}_{B\rightarrow A}^{\mathrm{P,twirled}}$ for $(\rho_A,\mathcal{N}_{A\rightarrow B})$ is such that 
$F_e(\rho_A,\mathcal{D}_{B\rightarrow A}^{\mathrm{P,twirled}}\circ\mathcal{N}_{A\rightarrow B})
=1$.
\item \emph{Saturation of data processing inequality for conditional entropy:} \\
$H(R|B)_{\sigma}=H(R|A)_{\rho}$.
\item \emph{Decoupling of $R$ from purifying system:} \\
If $\ket{\sigma}_{RBE}\in RBE$ is such that $\tr_E[\proj{\sigma}_{RBE}]=\sigma_{RB}$, then $\sigma_{RE}=\sigma_R\otimes \sigma_E$.
\end{enumerate}
\end{cor}

\begin{rem}[Optimality of Petz decoder] 
For the case where $\rho_A$ is a state whose spectrum contains exactly one non-zero element (e.g., the maximally mixed state), the optimality of the Petz decoder was recently studied in~\cite{li2024optimalityconditiontransposechannel}. 
They prove a necessary and sufficient condition for the Petz decoder to be optimal in terms of Kraus operators of $\mathcal{N}_{A\rightarrow B}$~\cite[Theorem~1]{li2024optimalityconditiontransposechannel}, and this equivalence holds regardless of whether perfect recovery is achievable.
\end{rem}

\subsection{Numerical comparison of decoders}\label{sec:numerics}
The preceding sections have focused on the performance of the SW, the Petz, and the twirled Petz decoder. 
In particular, we have shown that all three are universal decoders that achieve perfect recovery whenever possible (Corollary~\ref{cor:perfect}). 
Naturally, this raises the question of how these decoders compare when perfect recovery is unattainable. 
Does one of them generally perform better than the other two? 
By Corollary~\ref{cor:petz}, the Petz decoder generally outperforms the twirled Petz decoder. 
However, the analytical results presented in the previous sections do not address the question of how the performance of the Petz decoder compares to the (lower bound on the) performance of the SW decoder. 
We therefore ask the following questions.

\begin{question}\label{question1}
Which performs better: the SW decoder or the Petz decoder?
\end{question}

\begin{question}\label{question2}
How does the lower bound on the performance of the SW decoder in~\eqref{eq:duality-sw2} compare to the  performance of the Petz decoder, as expressed in~\eqref{eq:thm-duality}? 

Equivalently, how does 
$\widetilde{I}_{1/2}^{\uparrow\uparrow}(R:E)_{\mathcal{N}_{A\rightarrow E}^c(\proj{\rho}_{RA})}$ compare to 
$I_{1/2}^{\uparrow\downarrow}(R:E)_{\mathcal{N}_{A\rightarrow E}^c(\proj{\rho}_{RA})}$?
\end{question}

These questions will now be addressed for three settings, each corresponding to a different choice of $(\rho,\mathcal{N})$ and involving qubit systems with orthonormal basis $\{\ket{0},\ket{1}\}$:
\begin{itemize}
\item \emph{3-qubit bit-flip code, bit-flip channel.} 
The noisy channel is given by 
$\mathcal{N}\coloneqq\mathcal{M}^{\otimes 3}$ where $\mathcal{M}$ is the bit-flip channel defined as $\mathcal{M}(L)\coloneqq (1-p)L + pX L X$ for all linear operators $L$ acting on a qubit, where $X$ denotes the Pauli-$X$ gate defined as $X\coloneqq \ketbra{1}{0}+\ketbra{0}{1}$, and $p\in [0,1]$ is a fixed parameter. 
The initial state is defined as 
$\rho\coloneqq \frac{1}{2}(\proj{0_L}+\proj{1_L})$ where $\ket{0_L}\coloneqq \ket{000},\ket{1_L}\coloneqq \ket{111}$. 
\item \emph{LNCY 4-qubit code, amplitude damping channel.} 
The noisy channel is given by 
$\mathcal{N}\coloneqq\mathcal{M}^{\otimes 4}$ where $\mathcal{M}$ is the amplitude damping channel defined as
$\mathcal{M}(L)\coloneqq \sum_{i=0}^1K_iL K_i^\dagger $ for all linear operators $L$ acting on a qubit, 
where $K_0\coloneqq \proj{0}+\sqrt{1-p}\proj{1}$, $K_1\coloneqq \sqrt{p}\ketbra{0}{1}$, 
and $p\in [0,1]$ is a fixed parameter.  
The initial state is defined as 
$\rho\coloneqq \frac{1}{2}(\proj{0_L}+\proj{1_L})$ where $\ket{0_L}\coloneqq \frac{1}{\sqrt{2}}( \ket{0000}+\ket{1111}),\ket{1_L}\coloneqq \frac{1}{\sqrt{2}}( \ket{0011}+\ket{1100})$. 
This is the encoding originally proposed in~\cite{leung1997approximate}. 
\item \emph{5-qubit code, amplitude damping channel.} 
The noisy channel is given by 
$\mathcal{N}\coloneqq\mathcal{M}^{\otimes 5}$ where $\mathcal{M}$ is again the amplitude damping channel. 
The initial state is defined as 
$\rho\coloneqq \frac{1}{2}(\proj{0_L}+\proj{1_L})$ where $\ket{0_L}$ is defined as in~\cite[Eq.~(54)--(56)]{devitt2013quantum} (see also the original work~\cite{laflamme1996perfect}) and $\ket{1_L}\coloneqq X^{\otimes 5}\ket{0_L}$, where $X$ denotes the Pauli-$X$ gate. 
\end{itemize}

For each of these settings, the performance of the SW, the Petz, and the twirled Petz decoder has been computed as measured by the entanglement fidelity. 
The results are presented in Figure~\ref{fig:numerics}, where $F_e(\rho,\mathcal{D}\circ \mathcal{N})$ is plotted for $\mathcal{D}$ being the \emph{SW decoder} (green), the \emph{Petz decoder} (orange), or the \emph{twirled Petz decoder} (red, dashed). 
For illustrative purposes, the following additional curves are plotted in Figure~\ref{fig:numerics}: 
The \emph{upper bound} (black) is a plot of $F_e(\rho,\mathcal{D}^{\mathrm{P}}\circ \mathcal{N})^{1/2}$, which is an upper bound on the entanglement fidelity of an optimal decoder, see Remark~\ref{rem:optimal}. 
The curve labelled as \emph{optimal decoder} (blue) represents the entanglement fidelity of an optimal decoder, i.e., $\max_{\mathcal{D}} F_e(\rho,\mathcal{D}\circ \mathcal{N})$. 
This expression has been computed by solving the corresponding semidefinite program~\cite[Section~IV]{fletcher2007optimum}. 
To solve this semidefinite program for the 5-qubit code, we have taken advantage of the fact that the quantity of interest can be computed via a semidefinite program in a lower-dimensional setting~\cite[Section~V]{fletcher2007optimum}. 
The \emph{lower bound SW} (black, dashed) curve is a plot of $2^{I_2^{\downarrow}(\sigma_{RB}\| \sigma_R^{-1})}$, which is a lower bound on $F_e(\rho,\mathcal{D}^{\mathrm{SW}}\circ\mathcal{N})$, see~\eqref{eq:fidelity-sw}. 
The \emph{lower bound twirled Petz} (black, dotted) curve is a plot of $2^{D(\sigma_{RB}\| \sigma_R^{-1}\otimes \sigma_B)}$, which is a lower bound on $F_e(\rho,\mathcal{D}^{\mathrm{P,twirled}}\circ\mathcal{N})$, see~\eqref{eq:thm-bound3}.
The \emph{no decoder} (purple) curve is a plot of $F_e(\rho,\mathcal{D}\circ\mathcal{N})$ where $\mathcal{D}$ is taken to be the identity channel.

\begin{figure}
\includegraphics[height=0.23\textwidth]{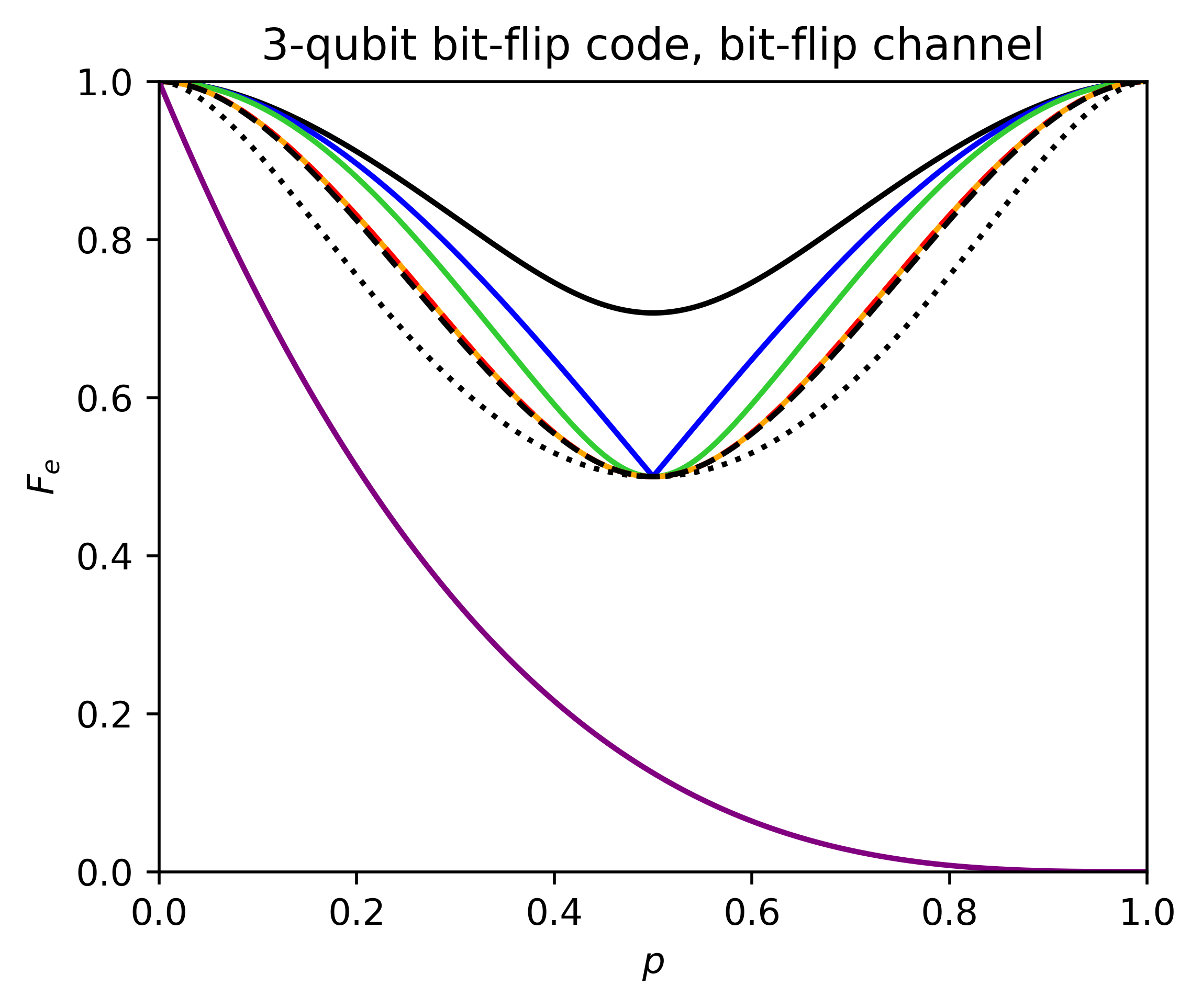}
\includegraphics[height=0.23\textwidth]{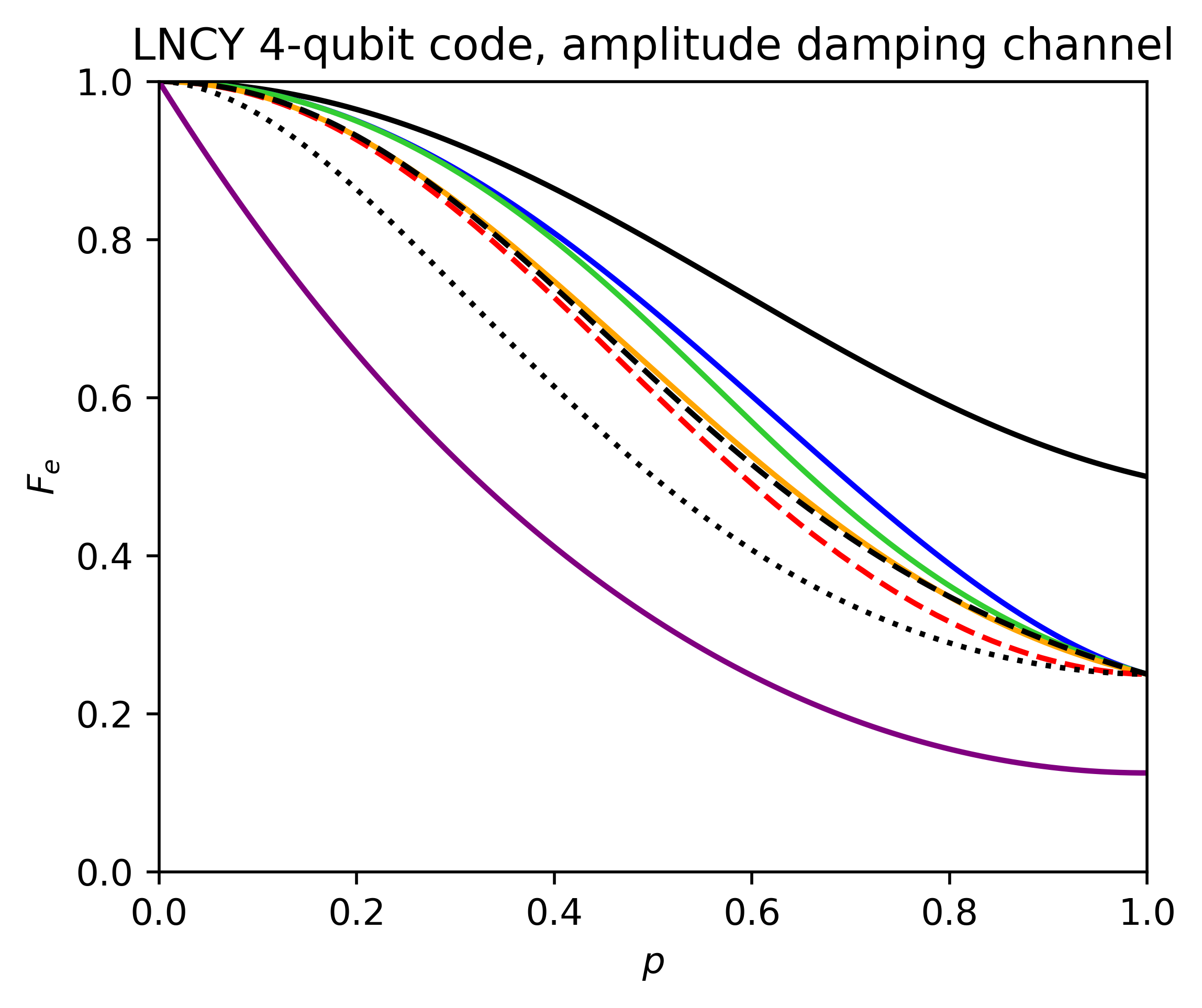}
\includegraphics[height=0.23\textwidth]{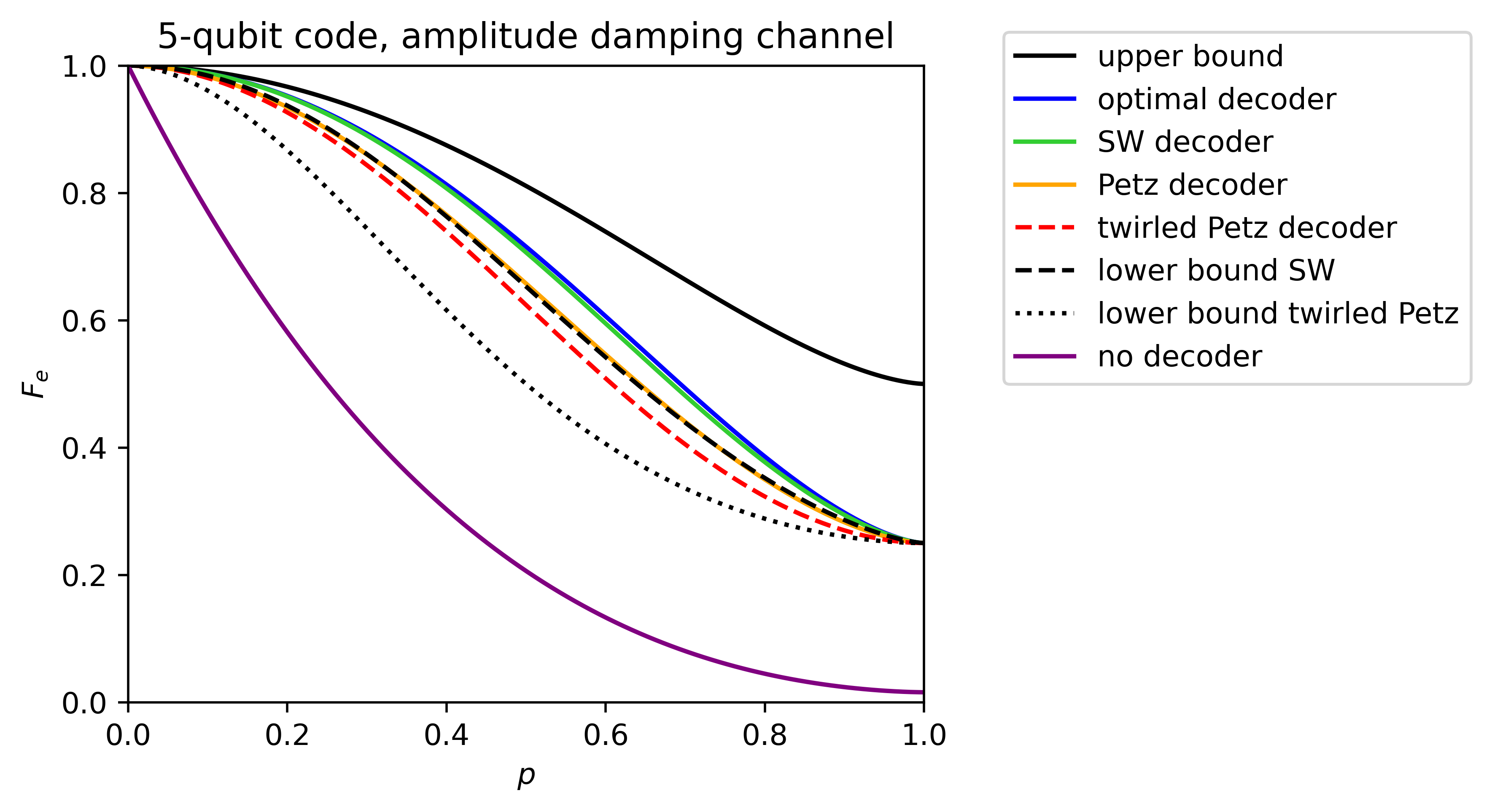}
\caption{Performance of SW, Petz, and twirled Petz decoder as measured by the entanglement fidelity $F_e$.}
\label{fig:numerics}
\end{figure}

The results presented in Figure~\ref{fig:numerics} imply the following qualitative statements. 
\begin{itemize}
\item \emph{3-qubit bit-flip code, bit-flip channel.} 
The SW decoder performs strictly better than the Petz decoder for any $p\in (0,1/2)\cup (1/2,1)$; for $p\in \{0,1/2,1\}$, their performances coincide. 
The Petz decoder and the twirled Petz decoder have the same entanglement fidelity. 
In comparison, the lower bound SW is slightly smaller if $p\in (0,1/2)\cup (1/2,1)$. 

\item \emph{LNCY 4-qubit code, amplitude damping channel.} 
The SW decoder performs strictly better than the Petz decoder for any $p\in (0,1)$; for $p\in \{0,1\}$, their performances coincide. 
The entanglement fidelity of the twirled Petz decoder is strictly smaller than that of the Petz decoder for any $p\in (0,1)$. 
The lower bound SW is slightly smaller or greater than the entanglement fidelity of the Petz decoder depending on the value of $p\in [0,1]$. 

\item \emph{5-qubit code, amplitude damping channel.} 
The SW decoder performs strictly better than the Petz decoder for any $p\in (0,1)$; for $p\in \{0,1\}$, their performances coincide. 
The entanglement fidelity of the twirled Petz decoder is strictly smaller than that of the Petz decoder for any $p\in (0,1)$. 
The lower bound SW is slightly smaller or greater than the entanglement fidelity of the Petz decoder depending on the value of $p\in [0,1]$. 
\end{itemize}
These conclusions allow us to answer the questions raised above as follows.

\begin{description}
    \item[Answer to Question~\ref{question1}] In all three settings, the SW decoder performs better than the Petz decoder, and in most cases it performs strictly better. 
    \item[Answer to Question~\ref{question2}] Depending on the setting and the parameter value $p\in [0,1]$, it is possible that 
$\widetilde{I}_{1/2}^{\uparrow\uparrow}(R:E)_{\mathcal{N}_{A\rightarrow E}^c(\proj{\rho}_{RA})}$ is greater or smaller than 
$I_{1/2}^{\uparrow\downarrow}(R:E)_{\mathcal{N}_{A\rightarrow E}^c(\proj{\rho}_{RA})}$. 
Therefore, these two quantities cannot be ordered by the same inequality for all possible $(\rho,\mathcal{N})$.
\end{description}

\section{Conclusion}\label{sec:discussion}
We defined two specific decoders for one-shot entanglement transmission, the Petz decoder and the twirled Petz decoder, and analyzed their performance as measured by the entanglement fidelity. 
The main results of this analysis were presented in Theorem~\ref{thm:petz}, accompanied by Corollary~\ref{cor:petz}. 
Theorem~\ref{thm:petz} asserts that the Petz decoder $\mathcal{D}^{\mathrm{P}}_{B\rightarrow A}$ for one-shot entanglement transmission of $\rho_A$ over a noisy channel $\mathcal{N}_{A\rightarrow B}$ satisfies
\begin{align}\label{eq:conclusion-petz1}
\log F_e(\rho_A,\mathcal{D}^{\mathrm{P}}_{B\rightarrow A}\circ\mathcal{N}_{A\rightarrow B})
=\widetilde{I}_2^\uparrow(\sigma_{RB}\| \sigma_R^{-1}),
\end{align}
where $\sigma_{RB}\coloneqq\mathcal{N}_{A\rightarrow B}(\proj{\rho}_{RA})$ and $\ket{\rho}_{RA}$ is a purification of $\rho_A$. 
\eqref{eq:conclusion-petz1} is a closed-form expression for the entanglement fidelity of the Petz decoder in terms of $\sigma_{RB}$. 
\eqref{eq:conclusion-petz1} can be expressed in an alternative way by duality, as shown in Corollary~\ref{cor:petz}: 
The entanglement fidelity of the Petz decoder is determined by the singly minimized Petz R\'enyi mutual information of order $1/2$ associated with a complementary channel $\mathcal{N}^c_{A\rightarrow E}$ as
\begin{align}\label{eq:conclusion-petz2}
\log F_e(\rho_A,\mathcal{D}^{\mathrm{P}}_{B\rightarrow A}\circ\mathcal{N}_{A\rightarrow B}) 
= -I_{1/2}^{\uparrow\downarrow}(R:E)_{\mathcal{N}^c_{A\rightarrow E}(\proj{\rho}_{RA})}.
\end{align}
Furthermore, Corollary~\ref{cor:petz} asserts that the twirled Petz decoder cannot outperform the Petz decoder. 
Thus, one-shot entanglement transmission is a task where the simpler Petz map suffices for recovery and even performs better. 
For other settings where the twirled Petz map is unnecessary and the Petz map is sufficient for recovery, see~\cite{alhambra2017dynamical,alhambra2018work,swingle2019recovery,cotler2019entanglement,chen2020entanglement}.

Our results on the performance of the Petz decoder can be compared with similar results on the performance of the SW decoder (Theorem~\ref{thm:sw}, Corollary~\ref{cor:sw}). 
For the SW decoder, it is not known whether the entanglement fidelity can be written as a closed-form expression of $\sigma_{RB}$. 
However, at least a lower bound on the performance can be proved that is a closed-form expression of $\sigma_{RB}$: According to Theorem~\ref{thm:sw}, we have
\begin{align}
\log F_e(\rho_A,\mathcal{D}^{\mathrm{SW}}_{B\rightarrow A}\circ\mathcal{N}_{A\rightarrow B})
\geq I_2^\downarrow(\sigma_{RB}\| \sigma_R^{-1}),
\end{align}
which is structurally similar to~\eqref{eq:conclusion-petz1}. 
This lower bound can be expressed in an alternative way by duality, as shown in Corollary~\ref{cor:sw}: 
The entanglement fidelity of the SW decoder is lower-bounded by the non-minimized sandwiched R\'enyi mutual information of order $1/2$ associated with a complementary channel $\mathcal{N}^c_{A\rightarrow E}$ as 
\begin{align}\label{eq:conclusion-sw2}
\log F_e(\rho_A,\mathcal{D}^{\mathrm{SW}}_{B\rightarrow A}\circ\mathcal{N}_{A\rightarrow B})
\geq -\widetilde{I}_{1/2}^{\uparrow\uparrow}(R:E)_{\mathcal{N}^c_{A\rightarrow E}(\proj{\rho}_{RA})},
\end{align}
which is structurally similar to~\eqref{eq:conclusion-petz2}. 

Our analytical results leave open the question of whether the SW decoder or the Petz decoder performs better. 
However, the numerical results in Section~\ref{sec:numerics} provide a partial answer: 
For all three settings examined there, the SW decoder outperformed the Petz decoder. 
At present, we are not aware of any example where the SW decoder performs strictly worse than the Petz decoder. 
Another question left open by our analytical results is how~\eqref{eq:conclusion-petz2} compares to the lower bound on the entanglement fidelity of the SW decoder in~\eqref{eq:conclusion-sw2}. 
Our numerical results show that this question does not have a universal answer: 
For some parameter ranges,~\eqref{eq:conclusion-petz2} was greater than the right-hand side of~\eqref{eq:conclusion-sw2}, while for others the opposite was true. 
Thus, although the SW decoder typically performs better than the Petz decoder, our analytical estimates of the performance of these decoders -- given by the right-hand sides of~\eqref{eq:conclusion-petz2} and~\eqref{eq:conclusion-sw2} -- appear to be comparably good.

The numerical results in Section~\ref{sec:numerics} also enable a comparison between the performance of the Petz decoder and that of an optimal decoder. 
The latter's performance can be formulated as a semidefinite program~\cite{fletcher2007optimum}, making its numerical computation relatively easy. 
In contrast, the Petz decoder has the advantage that its analytical form is known a priori, and its performance has a closed-form expression in terms of $\sigma_{RB}$, as shown in~\eqref{eq:conclusion-petz1}. 
These features are typically absent in an optimal decoder, but could be beneficial in certain applications.

Our result in~\eqref{eq:conclusion-petz2} provides an operational interpretation of the singly minimized Petz R\'enyi mutual information of order $1/2$ for states of the form $\mathcal{N}^c_{A\rightarrow E}(\proj{\rho}_{RA})$, i.e., states obtained from a bipartite pure state by acting on one of its parts with a CPTP map. 
For such states,~\eqref{eq:conclusion-petz2} shows that $I_{1/2}^{\uparrow\downarrow}(R:E)_{\mathcal{N}^c_{A\rightarrow E}(\proj{\rho}_{RA})}$ quantifies how well the initial state $\proj{\rho}_{RA}$ is recovered by the Petz decoder from a system $B$ that is the output of a channel $\mathcal{N}_{A\rightarrow B}$ complementary to $\mathcal{N}^c_{A\rightarrow E}$. 
Thus,~\eqref{eq:conclusion-petz2} reveals a connection between information recovery via the Petz map (\emph{left-hand side}) and a corresponding information measure (\emph{right-hand side}). 

Recently, the Petz map and the twirled Petz map have attracted interest in the research field of quantum theories of spacetime, as they provide a tool for entanglement wedge  reconstruction~\cite{vardhan2023petzrecoverysubsystemsconformal,bahiru2023explicit,cotler2019entanglement,chen2020entanglement}.
In this context, it has been conjectured that there exists a deeper connection between information recovery via the Petz map and the \emph{reflected entropy}~\cite{penington2020replicawormholesblackhole,akers2022page}. 
Our result in~\eqref{eq:conclusion-petz2} suggests that a modified version of this conjecture holds. 
As we will show in future work~\cite{burri2024minreflected}, the singly minimized Petz R\'enyi mutual information of order $1/2$ naturally serves as an upper bound for the min-reflected entropy. 
Consequently, by~\eqref{eq:conclusion-petz2}, it can be inferred that the min-reflected entropy provides an upper bound for the entanglement fidelity of the Petz decoder. 
This establishes a general relationship between information recovery via the Petz map and the \emph{min-reflected entropy}. 
It would be interesting to explore how tight this relationship is. 
More broadly, it remains an open question for future research whether there are even deeper connections between information recovery (via the Petz map) and the (min-)reflected entropy.

\begin{acknowledgments} 
I am grateful to Renato Renner for discussions and valuable comments on a draft of this work, 
and to Christophe Piveteau and Lukas Schmitt for discussions. 
This work was supported by the Swiss National Science Foundation via project No.\ 20QU-1\_225171
and the National Centre of Competence in Research SwissMAP, 
and the Quantum Center at ETH Zurich.
\end{acknowledgments}

\appendix 
\section{Construction of Schumacher-Westmoreland decoder}\label{app:sw-construction}
In this section, we summarize the construction of the SW decoder~\cite{schumacher2001approximate}. 
While our summary qualitatively follows the original description in~\cite{schumacher2001approximate}, it differs from the original presentation in a few places to make the construction more explicit and specific.
For instance, we have chosen to fix the dimensions of certain auxiliary systems for technical convenience, whereas the original work~\cite{schumacher2001approximate} leaves them unspecified. 
(Specifically, the dimension of $R$ will be the rank of $\rho_A$, and the dimension of $E$ will be $d_Ad_B$.)

Let $(\rho_A,\mathcal{N}_{A\rightarrow B})\in \mathcal{S}(A)\times \CPTP(A,B)$. 
Let $d$ be the rank of $\rho_A$, and let $R$ be a $d$-dimensional Hilbert space. 
Let $\ket{\rho}_{RA}\in RA$ be such that $\tr_R[\proj{\rho}_{RA}]=\rho_A$. 
Consider a Schmidt decomposition of $\ket{\rho}_{RA}$: 
Let $(\lambda_k)_{k\in [d]}\in [0,1]^{\times d}$ be a probability distribution, 
and let $\{\ket{r_k}_R\}_{k\in [d]},\{\ket{a_k}_A\}_{k\in [d_A]}$ be orthonormal bases for $R,A$ such that
\begin{align}\label{eq:rho_ra-schmidt}
\ket{\rho}_{RA}&=\sum_{k\in [d]} \sqrt{\lambda_k}\ket{r_k}_R\otimes\ket{a_k}_{A}.
\end{align}

Let $B'$ be isomorphic to $B$, let $E\coloneqq A B'$, and let $V\in \mathcal{L}(A,BE)$ be an isometry such that 
$\mathcal{N}_{A\rightarrow B}(X_A)=\tr_E[VX_AV^\dagger]$ for all $X_A\in \mathcal{L}(A)$. 
Let $\ket{\sigma}_{RBE}\coloneqq V\ket{\rho}_{RA}$. 
Consider a spectral decomposition of $\sigma_E$: 
Let $(\mu_l)_{l\in [d_E]}$ be a probability distribution and let $\{\ket{e_l}_{E}\}_{l\in [d_E]}$ be an orthonormal basis for $E$ such that
\begin{align}\label{eq:sigma_e-spectral}
\sigma_{E}&=\sum_{l\in [d_E]}\mu_l \proj{e_l}_{E}.
\end{align}

Let $A'$ be isomorphic to $A$, let $E'\coloneqq A'B$ (which is isomorphic to $E$), and let us define the following objects.
\begin{align}
\hat{\sigma}_{RE}
&\coloneqq \sigma_R\otimes\sigma_{E}
\label{eq:def-sigma-hat}
\\
\ket{\Omega}_{RER'E'}
&\coloneqq\sum_{k\in [d]}\sum_{l\in [d_E]} \ket{r_k}_R\otimes\ket{e_l}_{E}\otimes \ket{r_k}_{R'}\otimes\ket{e_l}_{E'}
\\
\ket{\Psi^{\sigma}}_{RER'E'}
&\coloneqq\sigma_{RE}^{\frac{1}{2}}\ket{\Omega}_{RER'E'}
\label{eq:def-phi-rho}
\\
M_{R'E'}
&\coloneqq \sum_{k,j\in [d]}\sum_{m,n\in [d_E]}
\bra{r_k}_R\otimes \bra{e_m}_{E}\hat{\sigma}_{RE}^{\frac{1}{2}}\sigma_{RE}^{\frac{1}{2}}\ket{r_j}_R\otimes \ket{e_n}_{E}\, 
\ketbra{r_j}{r_k}_{R'}\otimes \ketbra{e_n}{e_m}_{E'}
\label{eq:def-M-svd}
\\
\ket{\phi_{k,l}}_{R'E'}&\coloneqq \ket{r_k}_{R'}\otimes\ket{e_l}_{E'} 
\quad \forall k\in [d], l\in [d_E]
\label{eq:def-phi-kl}
\end{align}

Consider a singular value decomposition of $M_{R'E'}$: Let $\tilde{U}_{R'E'},\hat{U}_{R'E'}\in \mathcal{L}(R'E')$ be unitary and such that $M_{R'E'}=\hat{U}_{R'E'}\Sigma_{R'E'} \tilde{U}_{R'E'}$ for a suitable positive semidefinite operator $\Sigma_{R'E'}\in \mathcal{L}(R'E')$ that is diagonal in the orthonormal basis $\{\ket{r_k}_{R'}\otimes\ket{e_l}_{E'}\}_{k\in [d], l\in [d_E]}$. 
Let $U_{R'E'}\coloneqq \tilde{U}^\dagger_{R'E'}\hat{U}^\dagger_{R'E'}$. 

Let $\ket{0}_{R'A'}\in R'A'$ be an arbitrary but fixed unit vector. Then,
\begin{align}
\tr_{BR'A'}[\proj{\sigma}_{RBE}\otimes \proj{0}_{R'A'}]=\sigma_{RE}=\tr_{R'E'}[\proj{\Psi^\sigma}_{RER'E'}],
\end{align}
where the last equality follows from~\eqref{eq:def-phi-rho}. 
Therefore, there exists a unitary $W_{R'E'}\in \mathcal{L}(R'E')$ such that
\begin{align}\label{eq:def-w}
\ket{\Psi^\sigma}_{RER'E'}=W_{R'E'}\ket{\sigma}_{RBE}\otimes \ket{0}_{R'A'}.
\end{align} 

Let us define the following objects.
\begin{align}
\mathcal{D}'_{B\rightarrow R'E'}(X_B)&\coloneqq X_B \otimes \proj{0}_{R'A'}
\qquad\forall X_{B}\in \mathcal{L}(B)
\label{eq:def-d1}\\
K_l&\coloneqq \sum_{k\in [d]} \ket{a_k}_{A}\bra{\phi_{k,l}}_{R'E'}U_{R'E'}W_{R'E'}
\qquad \forall l\in [d_E]
\label{eq:def-kraus}\\
\mathcal{D}''_{R'E'\rightarrow A}(X_{R'E'})&\coloneqq \sum_{l\in [d_E]} K_l X_{R'E'} K_l^\dagger
\qquad\forall X_{R'E'}\in \mathcal{L}(R'E')
\label{eq:def-d2}
\end{align}
The SW decoder for $(\rho_A,\mathcal{N}_{A\rightarrow B})$ is then defined as 
\begin{align}\label{eq:d-sw}
\mathcal{D}_{B\rightarrow A}^{\mathrm{SW}}
&\coloneqq\mathcal{D}''_{R'E'\rightarrow A}\circ \mathcal{D}'_{B\rightarrow R'E'}.
\end{align}

\section{Proofs and remarks}
\subsection{Proof of Theorem~\ref{thm:sw}} \label{app:sw-proof}
\begin{proof}
\emph{Case 1: $d_R=\rank (\rho_A)$.} 
Then, $\ket{\rho}_{RA}$ has the same form as in the construction of the SW decoder. 
Consequently, all objects can be defined as in the construction of the SW decoder, see Section~\ref{app:sw-construction}. 
In addition, we define the following unit vectors.
\begin{align}
\ket{\Psi^{\hat{\sigma}}}_{RER'E'}
&\coloneqq\hat{\sigma}_{RE}^{\frac{1}{2}}\ket{\Omega}_{RER'E'}
=\sum_{k\in [d]}\sum_{l\in [d_E]}\sqrt{\lambda_k\mu_l}\ket{r_k}_R\otimes \ket{e_l}_E\otimes \ket{\phi_{k,l}}_{R'E'}
\label{eq:phi_hat_sigma}\\
\ket{\psi^{\hat{\sigma}}}_{RER'E'}
&\coloneqq W_{R'E'}^\dagger U_{R'E'}^\dagger \ket{\Psi^{\hat{\sigma}}}_{RER'E'}
\label{eq:psi_hat_sigma}
\end{align}
The second equality in~\eqref{eq:phi_hat_sigma} follows from~\eqref{eq:rho_ra-schmidt}, \eqref{eq:sigma_e-spectral}, and~\eqref{eq:def-sigma-hat} because $\sigma_R=\rho_R$. 
We have
\begin{align}
&\mathcal{D}_{R'E'\rightarrow A}''(\proj{\psi^{\hat{\sigma}}}_{RER'E'})
\\
&=\sum_{l\in [d_E]}\sum_{k\in [d]}\ket{a_k}_{A}\bra{\phi_{k,l}}_{R'E'}
\proj{\Psi^{\hat{\sigma}}}_{RER'E'}
\sum_{j\in [d]} \ket{\phi_{j,l}}_{R'E'} \bra{a_j}_{A}
\label{eq:d-1}\\
&=\sum_{l\in [d_E]} \sum_{k\in [d]}
\sqrt{\lambda_k\mu_l}
\ket{r_k}_R\otimes \ket{e_l}_E \otimes \ket{a_k}_A
\sum_{j\in [d]}
\sqrt{\lambda_j\mu_l}
\bra{r_j}_R\otimes \bra{e_l}_E \otimes \bra{a_j}_A
\label{eq:d-2}\\
&=\proj{\rho}_{RA} \otimes \sigma_E .
\label{eq:d-3}
\end{align}
\eqref{eq:d-1} follows from~\eqref{eq:def-kraus},~\eqref{eq:def-d2}, and~\eqref{eq:psi_hat_sigma}.
\eqref{eq:d-2} follows from~\eqref{eq:def-phi-kl} and~\eqref{eq:phi_hat_sigma}.
\eqref{eq:d-3} follows from~\eqref{eq:rho_ra-schmidt} and~\eqref{eq:sigma_e-spectral}. 

We are now ready to analyze the quantity of interest.
\begin{align}
F_e(\rho_A,\mathcal{D}_{B\rightarrow A}^{\mathrm{SW}}\circ\mathcal{N}_{A\rightarrow B})
&=F^2(\proj{\rho}_{RA},\mathcal{D}_{B\rightarrow A}^{\mathrm{SW}}\circ\mathcal{N}_{A\rightarrow B}(\proj{\rho}_{RA}))
\label{eq:proof-prop000}\\
&=F^2(\proj{\rho}_{RA},\mathcal{D}_{B\rightarrow A}^{\mathrm{SW}}(\sigma_{RB}))
\label{eq:proof-prop00}\\
&\geq F^2(\proj{\rho}_{RA} \otimes \sigma_E,\mathcal{D}_{B \rightarrow A}^{\mathrm{SW}} (\proj{\sigma}_{RBE}) )
\label{eq:proof-prop0}\\
&=F^2(\mathcal{D}''_{R'E'\rightarrow A}(\proj{\psi^{\hat{\sigma}}}_{RER'E'}) ,\mathcal{D}''_{R'E'\rightarrow A}\circ \mathcal{D}'_{B\rightarrow R'E'}(\proj{\sigma}_{RBE}) )
\label{eq:proof-prop1}\\
&\geq F^2(\proj{\psi^{\hat{\sigma}}}_{RER'E'},\mathcal{D}'_{B\rightarrow R'E'}(\proj{\sigma}_{RBE}) )
\label{eq:proof-prop2}\\
&=F^2(\proj{\psi^{\hat{\sigma}}}_{RER'E'} ,\proj{\sigma}_{RBE}\otimes\proj{0}_{R'A'} )
\label{eq:proof-prop3}\\
&=\lvert \bra{\psi^{\hat{\sigma}}}_{RER'E'}\ket{\sigma}_{RBE}\otimes\ket{0}_{R'A'} \rvert^2
\\
&=\lvert \bra{\Psi^{\hat{\sigma}}}_{RER'E'}1_{RE}\otimes U_{R'E'}\ket{\Psi^{\sigma}}_{RER'E'}\rvert^2
\label{eq:proof-prop4}\\
&=\lvert\bra{\Omega}_{RER'E'} \hat{\sigma}_{RE}^{1/2}\sigma_{RE}^{1/2} \otimes U_{R'E'} \ket{\Omega}_{RER'E'}\rvert^2
\label{eq:proof-prop5}\\
&=\lvert\tr[M_{R'E'}U_{R'E'}]\rvert^2
=\tr[\Sigma_{R'E'}]^2
=\rVert M_{R'E'}\rVert_1^2
=\lVert \hat{\sigma}_{RE}^{\frac{1}{2}}\sigma_{RE}^{\frac{1}{2}}\rVert_1^2
\label{eq:proof-prop6}\\
&=F^2(\hat{\sigma}_{RE},\sigma_{RE})
=F^2(\sigma_R\otimes \sigma_{E},\sigma_{RE})
\label{eq:proof-prop7}
\end{align}
\eqref{eq:proof-prop00} holds because $\sigma_{RB}=\tr_E[\proj{\sigma}_{RBE}]=\mathcal{N}_{A\rightarrow B}(\proj{\rho}_{RA})$.
\eqref{eq:proof-prop0} follows from the monotonicity of the fidelity under CPTP maps (and thus, under the partial trace over $E$). 
\eqref{eq:proof-prop1} follows from~\eqref{eq:d-sw} and~\eqref{eq:d-3}. 
\eqref{eq:proof-prop2} follows from the monotonicity of the fidelity under CPTP maps (and thus, under $\mathcal{D}''_{R'E'\rightarrow A}$). 
\eqref{eq:proof-prop3} follows from~\eqref{eq:def-d1}. 
\eqref{eq:proof-prop4} follows from~\eqref{eq:def-w} and~\eqref{eq:psi_hat_sigma}. 
\eqref{eq:proof-prop5} follows from~\eqref{eq:def-phi-rho} and~\eqref{eq:phi_hat_sigma}. 
\eqref{eq:proof-prop6} follows from~\eqref{eq:def-M-svd}. 
\eqref{eq:proof-prop7} follows from~\eqref{eq:def-sigma-hat}. 

By taking the logarithm, we can conclude that 
\begin{align}
\log F_e(\rho_A,\mathcal{D}_{B\rightarrow A}^{\mathrm{SW}}\circ\mathcal{N}_{A\rightarrow B})
&\geq \log F^2(\sigma_R\otimes \sigma_{E},\sigma_{RE})
\label{eq:proof-prop71}
\\
&=-\widetilde{D}_{1/2}(\sigma_{RE}\|  \sigma_R\otimes \sigma_{E})
=-\widetilde{I}_{1/2}^{\uparrow\uparrow}(R:E)_\sigma
\label{eq:proof-prop8}\\
&=I_2^\downarrow(\sigma_{RB}\| \sigma_R^{-1}).
\label{eq:proof-prop9}
\end{align}
\eqref{eq:proof-prop9} holds by a duality relation for the non-minimized sandwiched R\'enyi mutual information (or, equivalently, for the minimized generalized Petz R\'enyi mutual information)~\cite{hayashi2016correlation,burri2024doublyminimizedpetzrenyi,burri2024doublyminimizedsandwichedrenyi}. 

\emph{Case 2: $d_R>\rank (\rho_A)$.} 
Let $\tilde{R}$ be a Hilbert space whose dimension is $\rank(\rho_A)$. 
Let $\ket{\tilde{\rho}}_{\tilde{R}A}\in \tilde{R}A$ be such that $\tr_{\tilde{R}}[\proj{\tilde{\rho}}_{\tilde{R}A}]=\rho_A$. 
Since both $\ket{\rho}_{RA}$ and $\ket{\tilde{\rho}}_{\tilde{R}A}$ are purifications of $\rho_A$,  there exists an isometry $V\in \mathcal{L}(\tilde{R},R)$ such that 
$\ket{\rho}_{RA}=V\ket{\tilde{\rho}}_{\tilde{R}A}$. 
Let $\tilde{\sigma}_{\tilde{R}B}\coloneqq \mathcal{N}_{A\rightarrow B}(\proj{\tilde{\rho}}_{\tilde{R}A})$. 
Then,
\begin{align}
\sigma_{RB}
=\mathcal{N}_{A\rightarrow B}(\proj{\rho}_{RA})
=\mathcal{N}_{A\rightarrow B}(V\proj{\tilde{\rho}}_{\tilde{R}A}V^\dagger)
=V \tilde{\sigma}_{\tilde{R}B} V^\dagger .
\label{eq:proof-sigma12}
\end{align}
We can conclude that
\begin{align}\label{eq:proof-case2}
\log F_e(\rho_A,\mathcal{D}^{\mathrm{SW}}_{B\rightarrow A}\circ\mathcal{N}_{A\rightarrow B})
\geq I_2^\downarrow(\tilde{\sigma}_{\tilde{R}B}\| \tilde{\sigma}_{\tilde{R}}^{-1})
=I_2^\downarrow(V\tilde{\sigma}_{\tilde{R}B}V^\dagger \| V\tilde{\sigma}_{\tilde{R}}^{-1}V^\dagger)
= I_2^\downarrow(\sigma_{RB}\| \sigma_R^{-1}).
\end{align}
The inequality in~\eqref{eq:proof-case2} follows from case 1. 
The first equality in~\eqref{eq:proof-case2} follows from the invariance under local isometries of the minimized generalized Petz R\'enyi mutual information~\cite{burri2024doublyminimizedpetzrenyi}, 
and the second equality follows from~\eqref{eq:proof-sigma12}.
\end{proof}

\subsection{Proof of Corollary~\ref{cor:sw}} \label{app:sw-proof-cor}
\begin{proof}[Proof of~\eqref{eq:duality-sw}]
\eqref{eq:duality-sw} follows from Theorem~\ref{thm:sw} and the duality of the non-minimized sandwiched R\'enyi mutual information~\cite{hayashi2016correlation,burri2024doublyminimizedpetzrenyi,burri2024doublyminimizedsandwichedrenyi}.
\end{proof}
\begin{proof}[Proof of~\eqref{eq:duality-sw2}]
Since $\mathcal{N}^c_{A\rightarrow E}$ is a complementary channel to $\mathcal{N}_{A\rightarrow B}$, there exists an isometry $V\in \mathcal{L}(A,BE)$ such that 
$\mathcal{N}_{A\rightarrow B}(X_A)= \tr_E[VX_AV^\dagger] $ and 
$\mathcal{N}_{A\rightarrow E}^c(X_A)= \tr_B[VX_AV^\dagger] $ for all $X_A\in \mathcal{L}(A)$. 
Let $\ket{\tilde{\sigma}}_{RBE}\coloneqq V\ket{\rho}_{RA}$. 
Then
\begin{align}
\tilde{\sigma}_{RB}
&=\tr_E[\proj{\tilde{\sigma}}_{RBE}]
=\tr_E[V\proj{\rho}_{RA}V^\dagger]
=\mathcal{N}_{A\rightarrow B}(\proj{\rho}_{RA})
=\sigma_{RB},
\label{eq:proof-duality-sw2}\\
\tilde{\sigma}_{RE}
&=\tr_B[\proj{\tilde{\sigma}}_{RBE}]
=\tr_B[V\proj{\rho}_{RA}V^\dagger]
=\mathcal{N}_{A\rightarrow E}^c(\proj{\rho}_{RA}).
\label{eq:proof-duality-sw3}
\end{align} 
We can conclude that
\begin{align}
\log F_e(\rho_A,\mathcal{D}^{\mathrm{SW}}_{B\rightarrow A}\circ\mathcal{N}_{A\rightarrow B})
&\geq I_2^\downarrow(\sigma_{RB}\| \sigma_R^{-1})
\label{eq:proof-cor-sw1}
\\
&=I_2^\downarrow(\tilde{\sigma}_{RB}\| \tilde{\sigma}_R^{-1})
=-\widetilde{I}_{1/2}^{\uparrow\uparrow}(R:E)_{\tilde{\sigma}}
\label{eq:proof-cor-sw2}\\
&=-\widetilde{I}_{1/2}^{\uparrow\uparrow}(R:E)_{\mathcal{N}_{A\rightarrow E}^c(\proj{\rho}_{RA})} .
\label{eq:proof-cor-sw3}
\end{align}
\eqref{eq:proof-cor-sw1} follows from~\eqref{eq:duality-sw}. 
\eqref{eq:proof-cor-sw2} follows from~\eqref{eq:proof-duality-sw2} and the duality of the non-minimized sandwiched R\'enyi mutual information~\cite{hayashi2016correlation,burri2024doublyminimizedpetzrenyi,burri2024doublyminimizedsandwichedrenyi}. 
\eqref{eq:proof-cor-sw3} follows from~\eqref{eq:proof-duality-sw3}.
\end{proof}
\begin{proof}[Proof of~\eqref{eq:sw-original}]
Let $\ket{\sigma}_{RBE}\in RBE$ be such that $\tr_E[\proj{\sigma}_{RBE}]=\sigma_{RB}$. 
\begin{align}
\log
F_e(\rho_A,\mathcal{D}_{B\rightarrow A}^{\mathrm{SW}}\circ\mathcal{N}_{A\rightarrow B})
&\geq I_2^\downarrow(\sigma_{RB}\| \sigma_R^{-1})
=-\widetilde{I}_{1/2}^{\uparrow\uparrow}(R:E)_\sigma
=-\widetilde{D}_{1/2}(\sigma_{RE}\|  \sigma_R\otimes \sigma_{E})
\label{eq:proof-prop-1}\\
&\geq -\widetilde{D}_{1}(\sigma_{RE}\|  \sigma_R\otimes \sigma_{E})
\label{eq:proof-prop-2}\\
&=-D(\sigma_{RE}\|  \sigma_R\otimes \sigma_{E})
=-I(R:E)_\sigma=-H(R)_\sigma +H(R|E)_\sigma
\label{eq:proof-prop-20}\\
&=-H(R)_\sigma -H(R|B)_\sigma
=-H(R)_\sigma +I(R\,\rangle B)_\sigma
\label{eq:proof-prop-21}\\
&=D(\sigma_{RB}\| \sigma_R^{-1}\otimes \sigma_B)=-\varepsilon^\mathrm{SW}
\label{eq:proof-prop-22}\\
&=H(R|A)_\rho -H(R|B)_\sigma.
\label{eq:proof-prop-3}
\end{align}
\eqref{eq:proof-prop-1} follows from~\eqref{eq:duality-sw}. 
\eqref{eq:proof-prop-2} follows from the monotonicity of the sandwiched divergence in the R\'enyi order~\cite{mueller2013quantum}. 
\eqref{eq:proof-prop-21} follows from the duality of the conditional entropy~\cite{tomamichel2014relating}.
\eqref{eq:proof-prop-22} follows from~\eqref{eq:proof-prop-21} because $\log (\sigma_R^{-1})=-\log \sigma_R$. 
\eqref{eq:proof-prop-3} follows from~\eqref{eq:proof-prop-21} because $\proj{\rho}_{RA}$ is a pure state and $\rho_R=\sigma_R$.
\end{proof}

\subsection{More details on Remark~\ref{rem:original}}\label{app:original}
Based on the construction of the SW decoder (see Appendix~\ref{app:sw-construction}), 
the original work on the SW decoder~\cite{schumacher2001approximate} contains a proof of the following inequalities.
\begin{align}
F_e^{1/2}(\rho_A,\mathcal{D}_{B\rightarrow A}^{\mathrm{SW}}\circ\mathcal{N}_{A\rightarrow B})
&\geq F(\sigma_R\otimes \sigma_{E},\sigma_{RE})
\label{eq:original0}\\
&\geq 1-\frac{1}{2} \lVert \sigma_R\otimes \sigma_{E}-\sigma_{RE} \rVert_1
\label{eq:fuchs}\\
&\geq 1-\sqrt{\frac{\ln(2)}{2}D(\sigma_{RE}\| \sigma_R\otimes \sigma_E)} 
\label{eq:pinsker}\\
&=1-\sqrt{\frac{\ln(2)}{2}\varepsilon^{\mathrm{SW}}}
\label{eq:original-dual}\\
&\geq 1-\sqrt{\varepsilon^{\mathrm{SW}}}
\label{eq:original-duality}
\end{align}
\eqref{eq:original0} follows from similar arguments as in~\eqref{eq:proof-prop000}--\eqref{eq:proof-prop7}. 
\eqref{eq:fuchs} follows from the Fuchs-van de Graaf inequality~\cite{fuchs2004cryptographic}. 
\eqref{eq:pinsker} follows from the quantum Pinsker inequality~\cite{hiai1981sufficiency}. 
\eqref{eq:original-dual} follows from~\eqref{eq:proof-prop-21}, \eqref{eq:proof-prop-22}, and the definition of $\varepsilon^{\mathrm{SW}}$ in~\eqref{eq:def-epsilon}. 
The original proof therefore relies on a duality relation in~\eqref{eq:original-dual}, namely, the duality of the conditional entropy, see~\eqref{eq:proof-prop-21}. 
In contrast, our proof employs a different duality relation at an earlier stage of the proof, namely, the duality of the non-minimized sandwiched R\'enyi mutual information, see~\eqref{eq:proof-prop9}.

\subsection{Proof of Theorem~\ref{thm:petz}}\label{app:proof-thm-petz}

\begin{proof}[Proof of~\eqref{eq:thm-t}]
\begin{align}
&F_e(\rho_A,\mathcal{D}^{\mathrm{P},t}_{B\rightarrow A}\circ\mathcal{N}_{A\rightarrow B})
\label{eq:proof_pt00}\\
&=F^2(\proj{\rho}_{RA},\mathcal{D}^{\mathrm{P},t}_{B\rightarrow A}\circ\mathcal{N}_{A\rightarrow B}(\proj{\rho}_{RA}))
=\bra{\rho}_{RA}\mathcal{D}^{\mathrm{P},t}_{B\rightarrow A}\circ\mathcal{N}_{A\rightarrow B}(\proj{\rho}_{RA}) \ket{\rho}_{RA}
\label{eq:proof_pt0}\\
&=\bra{\rho}_{RA}\rho_A^{\frac{1}{2}(1-it)}\mathcal{N}^\dagger 
(\mathcal{N}(\rho_A)^{-\frac{1}{2}(1-it)} \mathcal{N}(\proj{\rho}_{RA}) \mathcal{N}(\rho_A)^{-\frac{1}{2}(1+it)} )
\rho_A^{\frac{1}{2}(1+it)}\ket{\rho}_{RA}
\label{eq:proof_pt1}\\
&=\bra{\rho}_{RA}\rho_R^{\frac{1}{2}(1-it)}\mathcal{N}^\dagger 
(\sigma_B^{-\frac{1}{2}(1-it)} \sigma_{RB} \sigma_B^{-\frac{1}{2}(1+it)} )
\rho_R^{\frac{1}{2}(1+it)}\ket{\rho}_{RA}
\\
&=\tr[\mathcal{N}^\dagger 
(\sigma_B^{-\frac{1}{2}(1-it)} \sigma_{RB} \sigma_B^{-\frac{1}{2}(1+it)} )
\rho_R^{\frac{1}{2}(1+it)}\proj{\rho}_{RA}\rho_R^{\frac{1}{2}(1-it)}]
\\
&=\tr[\sigma_B^{-\frac{1}{2}(1-it)} \sigma_{RB} \sigma_B^{-\frac{1}{2}(1+it)} \rho_R^{\frac{1}{2}(1+it)}
\mathcal{N}(\proj{\rho}_{RA})\rho_R^{\frac{1}{2}(1-it)}]
\\
&=\tr[\sigma_{RB} (\rho_R^{\frac{1}{2}(1+it)}\otimes \sigma_B^{-\frac{1}{2}(1+it)} )
\sigma_{RB}(\rho_R^{\frac{1}{2}(1-it)}\otimes \sigma_B^{-\frac{1}{2}(1-it)})]
\\
&=\lVert \sigma_{RB}^{\frac{1}{2}} (\rho_{R}^{\frac{1}{2}(1+it)}\otimes \sigma_{B}^{-\frac{1}{2}(1+it)}) \sigma_{RB}^{\frac{1}{2}} \rVert_2^2
=\lVert \sigma_{RB}^{\frac{1}{2}} (\sigma_{R}^{\frac{1}{2}(1+it)}\otimes \sigma_{B}^{-\frac{1}{2}(1+it)}) \sigma_{RB}^{\frac{1}{2}} \rVert_2^2
\label{eq:proof-thm-t}
\end{align}
\eqref{eq:proof_pt1} follows from~\eqref{eq:def_Rt}. 
\eqref{eq:proof-thm-t} holds because $\sigma_R=\rho_R$.
\end{proof}
\begin{proof}[Proof of~\eqref{eq:thm-petz}] 
$\mathcal{D}^{\mathrm{P}}_{B\rightarrow A}=\mathcal{P}_{\rho_A,\mathcal{N}_{A\rightarrow B}}
=\mathcal{R}_{\rho_A,\mathcal{N}_{A\rightarrow B}}^{0}=\mathcal{D}^{\mathrm{P},0}_{B\rightarrow A}$. 
Hence, the evaluation of~\eqref{eq:thm-t} for $t=0$ implies~\eqref{eq:thm-petz}.
\end{proof}
\begin{proof}[Proof of~\eqref{eq:thm-twirled}]
\begin{align}
F_e(\rho_A,\mathcal{D}^{\mathrm{P,twirled}}_{B\rightarrow A}\circ\mathcal{N}_{A\rightarrow B})
&=F^2(\proj{\rho}_{RA},\mathcal{D}^{\mathrm{P,twirled}}_{B\rightarrow A}\circ\mathcal{N}_{A\rightarrow B}(\proj{\rho}_{RA}))
\\
&=\bra{\rho}_{RA} \mathcal{D}^{\mathrm{P,twirled}}_{B\rightarrow A}\circ\mathcal{N}_{A\rightarrow B}(\proj{\rho}_{RA}) \ket{\rho}_{RA} 
\\
&=\int_{-\infty}^{\infty}\mathrm{d}t\, \beta_0(t) 
\bra{\rho}_{RA} \mathcal{D}^{\mathrm{P},t}_{B\rightarrow A}\circ\mathcal{N}_{A\rightarrow B}(\proj{\rho}_{RA}) \ket{\rho}_{RA} 
\\
&=\int_{-\infty}^{\infty}\mathrm{d}t\, \beta_0(t) F_e(\rho_A,\mathcal{D}^{\mathrm{P},t}_{B\rightarrow A}\circ\mathcal{N}_{A\rightarrow B}).
\label{eq:proof-twirled-t}
\end{align}
\eqref{eq:proof-twirled-t} follows from~\eqref{eq:proof_pt00} and~\eqref{eq:proof_pt0}. 
The equality in~\eqref{eq:thm-twirled} then follows from~\eqref{eq:thm-t}.
\end{proof}

\subsection{Proof of Corollary~\ref{cor:petz}}\label{app:proof-cor-petz}
In order to prove Corollary~\ref{cor:petz}, we will use the following lemma.
\begin{lem}\label{lem:trace}
Let $X,Y\in \mathcal{L}(A)$ be positive semidefinite. Then, for all $s,t\in \mathbb{R}$
\begin{align}\label{eq:lem}
0\leq \tr[XY^{s+it}XY^{s-it}]
\leq \tr[XY^sXY^s].
\end{align}
\end{lem}
\begin{proof}
The first inequality in~\eqref{eq:lem} follows from the assumption that $X$ and $Y$ are positive semidefinite. 
It remains to prove the second inequality in~\eqref{eq:lem}. 
Consider a spectral decomposition of $Y$: Let $\{\ket{e_j}\}_{j\in [d_A]}$ be an orthonormal basis for $A$ such that $Y=\sum_{j\in [d_A]}y_j\proj{e_j}_A$ where $y_j\in [0,\infty)$ are the eigenvalues of $Y$. 
Then
\begin{align}
\tr[XY^{s+it}XY^{s-it}]
&=\sum_{j,k\in [d_A]}y_j^{s+it}y_k^{s-it}\tr[X\proj{e_j} X\proj{e_k}]\\
&=\sum_{j,k\in [d_A]} y_j^{s+it}y_k^{s-it}\lvert\bra{e_j} X\ket{e_k}\rvert^2
\label{eq:proof-lemma}\\
&\leq \sum_{j,k\in [d_A]} \lvert y_j^{s+it}y_k^{s-it} \lvert \bra{e_j} X\ket{e_k}\rvert^2 \rvert\\
&= \sum_{j,k\in [d_A]}y_j^{s}y_k^{s}\lvert\bra{e_j} X\ket{e_k}\rvert^2
= \tr[XY^sXY^s].
\end{align}
\eqref{eq:proof-lemma} holds because $X$ is positive semidefinite, which entails that $X$ is self-adjoint ($X^\dagger=X$).
\end{proof}

We will now prove Corollary~\ref{cor:petz}.
\begin{proof}[Proof of (a)] 
\eqref{eq:thm-duality0} follows from~\eqref{eq:thm-petz} in Theorem~\ref{thm:petz} because 
$\widetilde{I}_2^\uparrow(\sigma_{RB}\| \sigma_R^{-1})=-I_{1/2}^{\uparrow\downarrow}(R:E)_\sigma$ by duality~\cite{hayashi2016correlation,burri2024doublyminimizedpetzrenyi,burri2024doublyminimizedsandwichedrenyi}. 
\eqref{eq:thm-duality} follows from~\eqref{eq:thm-duality0} by the same argument as in the proof of~\eqref{eq:duality-sw2}, see Appendix~\ref{app:sw-proof-cor}.
\end{proof}
\begin{proof}[Proof of (b)]
By Theorem~\ref{thm:petz}, the assertion is equivalent to the claim that for all $t\in\mathbb{R}$
\begin{align}
\tr[\sigma_{RB} (\sigma_R^{-1}\otimes \sigma_B )^{-\frac{1}{2}-i\frac{t}{2}}
\sigma_{RB}(\sigma_R^{-1}\otimes \sigma_B)^{-\frac{1}{2}+i\frac{t}{2}}]
&\leq \tr[\sigma_{RB} (\sigma_R^{-1}\otimes \sigma_B )^{-\frac{1}{2}}
\sigma_{RB}(\sigma_R^{-1}\otimes \sigma_B)^{-\frac{1}{2}}].
\end{align}
This inequality follows from Lemma~\ref{lem:trace}.
\end{proof}
\begin{proof}[Proof of (c)]
\eqref{eq:thm-bound1} follows from~\eqref{eq:proof-twirled-t} and~(b). 
It remains to prove~\eqref{eq:thm-bound3}. 
For any $\tau_R\in \mathcal{S}(R)$ such that $\rho_R\ll\tau_R$ and $\tau_R\ll \rho_R$
\begin{align}
&\log F_e(\rho_A,\mathcal{D}^\mathrm{P,twirled}_{B\rightarrow A}\circ\mathcal{N}_{A\rightarrow B})
\\
&=\log F^2(\proj{\rho}_{RA},(\mathcal{I}_{R\rightarrow R}\otimes \mathcal{R}_{\rho_A,\mathcal{N}_{A\rightarrow B}}\circ\mathcal{N}_{A\rightarrow B})(\proj{\rho}_{RA}))\\
&=\log F^2(\proj{\rho}_{RA},\mathcal{R}_{\rho_A\otimes \tau_R,\mathcal{N}_{A\rightarrow B}\otimes \mathcal{I}_{R\rightarrow R}}\circ(\mathcal{N}_{A\rightarrow B}\otimes \mathcal{I}_{R\rightarrow R})(\proj{\rho}_{RA}))\\
&\geq -D(\proj{\rho}_{RA}\| \tau_R\otimes \rho_A ) 
+ D(\mathcal{N}_{A\rightarrow B}(\proj{\rho}_{RA})\| \tau_R\otimes \mathcal{N}_{A\rightarrow B}(\rho_A) ) 
\label{eq:proof-c}\\
&=-D(\proj{\rho}_{RA}\| \tau_R\otimes\rho_A) + D(\sigma_{RB}\| \tau_R\otimes \sigma_B ) .
\end{align}
\eqref{eq:proof-c} follows from~\eqref{eq:d-recovery}. 
Taking the supremum over all such $\tau_R$, we find
\begin{align}
\log F_e(\rho_A,\mathcal{D}^\mathrm{P,twirled}_{B\rightarrow A}\circ\mathcal{N}_{A\rightarrow B})
&\geq \sup_{\substack{\tau_R\in \mathcal{S}(R):\\ \rho_R\ll\tau_R, \\ \tau_R\ll \rho_R}} 
[-D(\proj{\rho}_{RA}\| \tau_R\otimes\rho_A) + D(\sigma_{RB}\| \tau_R\otimes\sigma_B )]
\label{eq:proof-iii-1}\\
&\geq -I(R:A)_\rho +I(R:B)_\sigma
=H(R|A)_\rho-H(R|B)_\sigma
\label{eq:proof-iii-2}\\
&=D(\sigma_{RB}\| \sigma_R^{-1}\otimes \sigma_B).
\label{eq:proof-iii-3}
\end{align}
\eqref{eq:proof-iii-2} follows from~\eqref{eq:proof-iii-1} by choosing $\tau_R\coloneqq\rho_R=\sigma_R$. 
\eqref{eq:proof-iii-3} follows from~\eqref{eq:def-epsilon}.
\end{proof}

\subsection{Proof of Corollary~\ref{cor:perfect}}\label{app:corollary}
We will prove the equivalence of the statements by proving certain implications. 
\begin{proof}[(a) $\Rightarrow$ (f)]
Suppose~(a) holds, i.e., there exists $\mathcal{D}_{B\rightarrow A}\in \CPTP(B,A)$ such that 
$F_e(\rho_A,\mathcal{D}_{B\rightarrow A}\circ\mathcal{N}_{A\rightarrow B})
=1$. 
Then, $F(\proj{\rho}_{RA},\mathcal{D}_{B\rightarrow A}\circ\mathcal{N}_{A\rightarrow B}(\proj{\rho}_{RA}))
=1$, so
\begin{align}\label{eq:proof-cor1}
\proj{\rho}_{RA}=\mathcal{D}_{B\rightarrow A}\circ\mathcal{N}_{A\rightarrow B}(\proj{\rho}_{RA}).
\end{align}
Let $\ket{\sigma}_{RBE}\in RBE$ be such that $\tr_E[\proj{\sigma}_{RBE}]=\sigma_{RB}$. Then,
\begin{align}
0\leq I(R:E)_\sigma 
=H(R)_\sigma - H(R|E)_\sigma
&=H(R)_\rho +H(R|B)_\sigma 
\label{eq:proof-cor20}\\
&=-H(R|A)_\rho +H(R|B)_\sigma 
\label{eq:proof-cor2}\\
&=-H(R|A)_{\mathcal{D}_{B\rightarrow A}(\sigma_{RB}) } +H(R|B)_\sigma 
\leq 0
\label{eq:proof-cor3}
\end{align}
\eqref{eq:proof-cor20} follows from the duality of the conditional entropy. 
\eqref{eq:proof-cor2} holds because $\proj{\rho}_{RA}$ is a pure state. 
The equality in~\eqref{eq:proof-cor3} follows from~\eqref{eq:proof-cor1}. 
The inequality in~\eqref{eq:proof-cor3} follows from the data processing inequality for the conditional entropy. 
We can conclude that $I(R:E)_\sigma=0$, which implies that $\sigma_{RE}=\sigma_R\otimes \sigma_E$.
\end{proof}
\begin{proof}[(f) $\Rightarrow$ (e)] 
Let $\ket{\sigma}_{RBE}\in RBE$ be such that $\tr_E[\proj{\sigma}_{RBE}]=\sigma_{RB}$. 
Suppose~(f) holds, i.e., $\sigma_{RE}=\sigma_R\otimes \sigma_E$. Then,
\begin{align}
H(R|B)_\sigma =-H(R|E)_\sigma 
&=-H(R)_\sigma
\label{eq:proof_cor1}\\
&=-H(R)_\rho
=H(R|A)_\rho .
\label{eq:proof_cor2}
\end{align}
The first equality in~\eqref{eq:proof_cor1} follows from the duality of the conditional entropy, 
and the second equality follows from $\sigma_{RE}=\sigma_R\otimes \sigma_E$. 
The first equality in~\eqref{eq:proof_cor2} follows from $\rho_R=\sigma_R$, 
and the second equality follows from the fact that $\proj{\rho}_{RA}$ is a pure state.
\end{proof}
\begin{proof}[(e) $\Rightarrow$ (b) $\land$ (c) $\land$ (d)]
Suppose~(e) holds, i.e., $H(R|B)_\sigma=H(R|A)_\rho$. Then, 
$\varepsilon^{\mathrm{SW}}\coloneqq H(R|B)_{\sigma}-H(R|A)_{\rho}=0$. 
By Corollary~\ref{cor:sw}~(b), 
\begin{align}
0\geq \log F_e(\rho_A,\mathcal{D}^{\mathrm{SW}}_{B\rightarrow A}\circ\mathcal{N}_{A\rightarrow B})
\geq -\varepsilon^{\mathrm{SW}}=0,
\end{align}
which implies that 
$F_e(\rho_A,\mathcal{D}^{\mathrm{SW}}_{B\rightarrow A}\circ\mathcal{N}_{A\rightarrow B})=1$, i.e., (b).

By Corollary~\ref{cor:petz}~(c) and Remark~\ref{rem:sw-petz}, 
\begin{align}
0&\geq \log F_e(\rho_A,\mathcal{D}^{\mathrm{P}}_{B\rightarrow A}\circ\mathcal{N}_{A\rightarrow B})
\\
&\geq \log F_e(\rho_A,\mathcal{D}^{\mathrm{P,twirled}}_{B\rightarrow A}\circ\mathcal{N}_{A\rightarrow B})
\geq -\varepsilon^{\mathrm{SW}}=0,
\end{align}
which implies that 
$F_e(\rho_A,\mathcal{D}^{\mathrm{P}}_{B\rightarrow A}\circ\mathcal{N}_{A\rightarrow B})=1=F_e(\rho_A,\mathcal{D}^{\mathrm{P,twirled}}_{B\rightarrow A}\circ\mathcal{N}_{A\rightarrow B})$, i.e., (c) and (d).
\end{proof}
\begin{proof}[(b) $\lor$ (c) $\lor$ (d) $\Rightarrow$ (a)]
This implication is trivial because $\mathcal{D}_{B\rightarrow A}$ from (a) can be chosen to be 
$\mathcal{D}_{B\rightarrow A}^{\mathrm{SW}},\mathcal{D}_{B\rightarrow A}^{\mathrm{P}},$ or $\mathcal{D}_{B\rightarrow A}^{\mathrm{P,twirled}}$, respectively. 
\end{proof}

\bibliographystyle{arxiv_fullname}
\bibliography{bibfile}

\begin{thebibliography}{10}

\bibitem{schumacher1996sending}
Benjamin Schumacher.
\newblock Sending entanglement through noisy quantum channels.
\newblock {\em Physical Review A}, 54:2614--2628, 1996.
\newblock
  \texttt{\href{http://dx.doi.org/10.1103/PhysRevA.54.2614}{DOI:\,10.1103/PhysRevA.54.2614}}.

\bibitem{schumacher1996quantum}
Benjamin Schumacher and Michael~A. Nielsen.
\newblock Quantum data processing and error correction.
\newblock {\em Physical Review A}, 54(4):2629--2635, 1996.
\newblock
  \texttt{\href{http://dx.doi.org/10.1103/PhysRevA.54.2629}{DOI:\,10.1103/PhysRevA.54.2629}}.

\bibitem{schumacher2001approximate}
Benjamin Schumacher and Michael~D. Westmoreland.
\newblock Approximate quantum error correction, 2001.
\newblock
  \texttt{\href{http://dx.doi.org/10.48550/arXiv.quant-ph/0112106}{DOI:\,10.48550/arXiv.quant-ph/0112106}}.

\bibitem{fletcher2007optimum}
Andrew~S. Fletcher, Peter~W. Shor, and Moe~Z. Win.
\newblock Optimum quantum error recovery using semidefinite programming.
\newblock {\em Physical Review A}, 75:012338, 2007.
\newblock
  \texttt{\href{http://dx.doi.org/10.1103/PhysRevA.75.012338}{DOI:\,10.1103/PhysRevA.75.012338}}.

\bibitem{reimpell2005iterative}
Michael Reimpell and Reinhard~F. Werner.
\newblock Iterative optimization of quantum error correcting codes.
\newblock {\em Physical Review Letters}, 94(8), 2005.
\newblock
  \texttt{\href{http://dx.doi.org/10.1103/PhysRevLett.94.080501}{DOI:\,10.1103/PhysRevLett.94.080501}}.

\bibitem{reimpell2006commentoptimumquantumerror}
Michael Reimpell, Reinhard~F. Werner, and Koenraad Audenaert.
\newblock {Comment on ``Optimum Quantum Error Recovery using Semidefinite
  Programming''}, 2006.
\newblock
  \texttt{\href{http://dx.doi.org/10.48550/arXiv.quant-ph/0606059}{DOI:\,10.48550/arXiv.quant-ph/0606059}}.

\bibitem{datta2013one}
Nilanjana Datta and Min-Hsiu Hsieh.
\newblock One-shot entanglement-assisted quantum and classical communication.
\newblock {\em IEEE Transactions on Information Theory}, 59(3):1929–1939,
  2013.
\newblock
  \texttt{\href{http://dx.doi.org/10.1109/TIT.2012.2228737}{DOI:\,10.1109/TIT.2012.2228737}}.

\bibitem{beigi2016decoding}
Salman Beigi, Nilanjana Datta, and Felix Leditzky.
\newblock {Decoding quantum information via the Petz recovery map}.
\newblock {\em Journal of Mathematical Physics}, 57(8):082203, 2016.
\newblock
  \texttt{\href{http://dx.doi.org/10.1063/1.4961515}{DOI:\,10.1063/1.4961515}}.

\bibitem{petz1986sufficient}
D\'{e}nes Petz.
\newblock {Sufficient subalgebras and the relative entropy of states of a von
  Neumann algebra}.
\newblock {\em Communications in Mathematical Physics}, 105(1):123--131, 1986.
\newblock
  \texttt{\href{http://dx.doi.org/10.1007/BF01212345}{DOI:\,10.1007/BF01212345}}.

\bibitem{petz1988sufficiency}
D\'{e}nes Petz.
\newblock {Sufficiency of channels over von Neumann algebras}.
\newblock {\em The Quarterly Journal of Mathematics}, 39(1):97--108, 1988.
\newblock
  \texttt{\href{http://dx.doi.org/10.1093/qmath/39.1.97}{DOI:\,10.1093/qmath/39.1.97}}.

\bibitem{ohya1993quantum}
Masanori Ohya and D{\'e}nes Petz.
\newblock {\em {Quantum Entropy and Its Use}}.
\newblock Springer, Berlin, 1993.

\bibitem{petz2003monotonicity}
D\'{e}nes Petz.
\newblock Monotonicity of quantum relative entropy revisited.
\newblock {\em Reviews in Mathematical Physics}, 15(01):79--91, 2003.
\newblock
  \texttt{\href{http://dx.doi.org/10.1142/S0129055X03001576}{DOI:\,10.1142/S0129055X03001576}}.

\bibitem{junge2018universal}
Marius Junge, Renato Renner, David Sutter, Mark~M. Wilde, and Andreas Winter.
\newblock Universal recovery maps and approximate sufficiency of quantum
  relative entropy.
\newblock {\em Annales Henri Poincar{\'{e}}}, 19(10):2955--2978, 2018.
\newblock
  \texttt{\href{http://dx.doi.org/10.1007/s00023-018-0716-0}{DOI:\,10.1007/s00023-018-0716-0}}.

\bibitem{leung1997approximate}
Debbie~W. Leung, Michael~A. Nielsen, Isaac~L. Chuang, and Yoshihisa Yamamoto.
\newblock Approximate quantum error correction can lead to better codes.
\newblock {\em Physical Review A}, 56(4):2567--2573, 1997.
\newblock
  \texttt{\href{http://dx.doi.org/10.1103/PhysRevA.56.2567}{DOI:\,10.1103/PhysRevA.56.2567}}.

\bibitem{laflamme1996perfect}
Raymond Laflamme, Cesar Miquel, Juan~Pablo Paz, and Wojciech~Hubert Zurek.
\newblock Perfect quantum error correcting code.
\newblock {\em Physical Review Letters}, 77:198--201, 1996.
\newblock
  \texttt{\href{http://dx.doi.org/10.1103/PhysRevLett.77.198}{DOI:\,10.1103/PhysRevLett.77.198}}.

\bibitem{petz1986quasi}
D\'enes Petz.
\newblock Quasi-entropies for finite quantum systems.
\newblock {\em Reports on Mathematical Physics}, 23(1):57--65, 1986.
\newblock
  \texttt{\href{http://dx.doi.org/10.1016/0034-4877(86)90067-4}{DOI:\,10.1016/0034-4877(86)90067-4}}.

\bibitem{lin2015investigating}
Simon~M. Lin and Marco Tomamichel.
\newblock {Investigating properties of a family of quantum R\'enyi
  divergences}.
\newblock {\em Quantum Information Processing}, 14(4):1501--1512, 2015.
\newblock
  \texttt{\href{http://dx.doi.org/10.1007/s11128-015-0935-y}{DOI:\,10.1007/s11128-015-0935-y}}.

\bibitem{tomamichel2016quantum}
Marco Tomamichel.
\newblock {\em {Quantum Information Processing with Finite Resources}}.
\newblock Springer, 2016.
\newblock
  \texttt{\href{http://dx.doi.org/10.1007/978-3-319-21891-5}{DOI:\,10.1007/978-3-319-21891-5}}.

\bibitem{mueller2013quantum}
Martin Müller-Lennert, Fr{\'{e}}d{\'{e}}ric Dupuis, Oleg Szehr, Serge Fehr,
  and Marco Tomamichel.
\newblock {On quantum R{\'{e}}nyi entropies: A new generalization and some
  properties}.
\newblock {\em Journal of Mathematical Physics}, 54(12):122203, 2013.
\newblock
  \texttt{\href{http://dx.doi.org/10.1063/1.4838856}{DOI:\,10.1063/1.4838856}}.

\bibitem{wilde2014strong}
Mark~M. Wilde, Andreas Winter, and Dong Yang.
\newblock {Strong Converse for the Classical Capacity of Entanglement-Breaking
  and Hadamard Channels via a Sandwiched R\'enyi Relative Entropy}.
\newblock {\em Communications in Mathematical Physics}, 331:593--622, 2014.
\newblock
  \texttt{\href{http://dx.doi.org/10.1007/s00220-014-2122-x}{DOI:\,10.1007/s00220-014-2122-x}}.

\bibitem{hayashi2016correlation}
Masahito Hayashi and Marco Tomamichel.
\newblock {Correlation detection and an operational interpretation of the
  R\'enyi mutual information}.
\newblock {\em Journal of Mathematical Physics}, 57(102201), 2016.
\newblock
  \texttt{\href{http://dx.doi.org/10.1063/1.4964755}{DOI:\,10.1063/1.4964755}}.

\bibitem{burri2024doublyminimizedpetzrenyi}
Laura Burri.
\newblock {Doubly minimized Petz R\'enyi mutual information: Properties and
  operational interpretation from direct exponent}, 2024.
\newblock
  \texttt{\href{http://dx.doi.org/10.48550/arXiv.2406.01699}{DOI:\,10.48550/arXiv.2406.01699}}.

\bibitem{gupta2014multiplicativity}
Manish~K. Gupta and Mark~M. Wilde.
\newblock {Multiplicativity of Completely Bounded $p$-Norms Implies a Strong
  Converse for Entanglement-Assisted Capacity}.
\newblock {\em Communications in Mathematical Physics}, 334(2):867--887, 2014.
\newblock
  \texttt{\href{http://dx.doi.org/10.1007/s00220-014-2212-9}{DOI:\,10.1007/s00220-014-2212-9}}.

\bibitem{burri2024doublyminimizedsandwichedrenyi}
Laura Burri.
\newblock {Doubly minimized sandwiched R\'enyi mutual information: Properties
  and operational interpretation from strong converse exponent}, 2024.
\newblock
  \texttt{\href{http://dx.doi.org/10.48550/arXiv.2406.03213}{DOI:\,10.48550/arXiv.2406.03213}}.

\bibitem{lindblad1975completely}
G\"oran Lindblad.
\newblock Completely positive maps and entropy inequalities.
\newblock {\em Communications in Mathematical Physics}, 40:147--151, 1975.
\newblock
  \texttt{\href{http://dx.doi.org/10.1007/BF01609396}{DOI:\,10.1007/BF01609396}}.

\bibitem{uhlmann1977relative}
Armin Uhlmann.
\newblock {Relative Entropy and the Wigner-Yanase-Dyson-Lieb Concavity in an
  Interpolation Theory}.
\newblock {\em Communications in Mathematical Physics}, 54:21, 1977.
\newblock
  \texttt{\href{http://dx.doi.org/10.1007/BF01609834}{DOI:\,10.1007/BF01609834}}.

\bibitem{wilde2015recoverability}
Mark~M. Wilde.
\newblock Recoverability in quantum information theory.
\newblock {\em Proceedings of the Royal Society A: Mathematical, Physical and
  Engineering Sciences}, 471(2182):20150338, 2015.
\newblock
  \texttt{\href{http://dx.doi.org/10.1098/RSPA.2015.0338}{DOI:\,10.1098/RSPA.2015.0338}}.

\bibitem{sutter2016strengthened}
David Sutter, Marco Tomamichel, and Aram~W. Harrow.
\newblock {Strengthened Monotonicity of Relative Entropy via Pinched Petz
  Recovery Map}.
\newblock {\em IEEE Transactions on Information Theory}, 62(5):2907--2913,
  2016.
\newblock
  \texttt{\href{http://dx.doi.org/10.1109/TIT.2016.2545680}{DOI:\,10.1109/TIT.2016.2545680}}.

\bibitem{barnum2002reversing}
Howard Barnum and Emanuel Knill.
\newblock Reversing quantum dynamics with near-optimal quantum and classical
  fidelity.
\newblock {\em Journal of Mathematical Physics}, 43(5):2097--2106, 2002.
\newblock
  \texttt{\href{http://dx.doi.org/10.1063/1.1459754}{DOI:\,10.1063/1.1459754}}.

\bibitem{ng2010simple}
Hui~Khoon Ng and Prabha Mandayam.
\newblock Simple approach to approximate quantum error correction based on the
  transpose channel.
\newblock {\em Physical Review A}, 81(6), 2010.
\newblock
  \texttt{\href{http://dx.doi.org/10.1103/PhysRevA.81.062342}{DOI:\,10.1103/PhysRevA.81.062342}}.

\bibitem{li2024optimalityconditiontransposechannel}
Bikun Li, Zhaoyou Wang, Guo Zheng, and Liang Jiang.
\newblock {Optimality Condition for the Transpose Channel}, 2024.
\newblock
  \texttt{\href{http://dx.doi.org/10.48550/arXiv.2410.23622}{DOI:\,10.48550/arXiv.2410.23622}}.

\bibitem{devitt2013quantum}
Simon~J. Devitt, William~J. Munro, and Kae Nemoto.
\newblock Quantum error correction for beginners.
\newblock {\em Reports on Progress in Physics}, 76(7):076001, 2013.
\newblock
  \texttt{\href{http://dx.doi.org/10.1088/0034-4885/76/7/076001}{DOI:\,10.1088/0034-4885/76/7/076001}}.

\bibitem{alhambra2017dynamical}
Álvaro~M. Alhambra and Mischa~P. Woods.
\newblock {Dynamical maps, quantum detailed balance, and the Petz recovery
  map}.
\newblock {\em Physical Review A}, 96(2), 2017.
\newblock
  \texttt{\href{http://dx.doi.org/10.1103/PhysRevA.96.022118}{DOI:\,10.1103/PhysRevA.96.022118}}.

\bibitem{alhambra2018work}
Álvaro~M. Alhambra, Stephanie Wehner, Mark~M. Wilde, and Mischa~P. Woods.
\newblock Work and reversibility in quantum thermodynamics.
\newblock {\em Physical Review A}, 97(6), 2018.
\newblock
  \texttt{\href{http://dx.doi.org/10.1103/PhysRevA.97.062114}{DOI:\,10.1103/PhysRevA.97.062114}}.

\bibitem{swingle2019recovery}
Brian~G. Swingle and Yixu Wang.
\newblock {Recovery map for fermionic Gaussian channels}.
\newblock {\em Journal of Mathematical Physics}, 60(7), 2019.
\newblock
  \texttt{\href{http://dx.doi.org/10.1063/1.5093326}{DOI:\,10.1063/1.5093326}}.

\bibitem{cotler2019entanglement}
Jordan Cotler, Patrick Hayden, Geoffrey Penington, Grant Salton, Brian Swingle,
  and Michael Walter.
\newblock {Entanglement Wedge Reconstruction via Universal Recovery Channels}.
\newblock {\em Physical Review X}, 9:031011, 2019.
\newblock
  \texttt{\href{http://dx.doi.org/10.1103/PhysRevX.9.031011}{DOI:\,10.1103/PhysRevX.9.031011}}.

\bibitem{chen2020entanglement}
Chi-Fang Chen, Geoffrey Penington, and Grant Salton.
\newblock {Entanglement Wedge Reconstruction using the Petz map}.
\newblock {\em Journal of High Energy Physics}, 2020(1), 2020.
\newblock
  \texttt{\href{http://dx.doi.org/10.1007/JHEP01(2020)168}{DOI:\,10.1007/JHEP01(2020)168}}.

\bibitem{vardhan2023petzrecoverysubsystemsconformal}
Shreya Vardhan, Annie~Y. Wei, and Yijian Zou.
\newblock Petz recovery from subsystems in conformal field theory, 2023.
\newblock
  \texttt{\href{http://dx.doi.org/10.48550/arXiv.2307.14434}{DOI:\,10.48550/arXiv.2307.14434}}.

\bibitem{bahiru2023explicit}
Eyoab Bahiru and Niloofar Vardian.
\newblock {Explicit reconstruction of the entanglement wedge via the Petz map}.
\newblock {\em Journal of High Energy Physics}, 2023(7), 2023.
\newblock
  \texttt{\href{http://dx.doi.org/10.1007/JHEP07(2023)025}{DOI:\,10.1007/JHEP07(2023)025}}.

\bibitem{penington2020replicawormholesblackhole}
Geoff Penington, Stephen~H. Shenker, Douglas Stanford, and Zhenbin Yang.
\newblock Replica wormholes and the black hole interior.
\newblock {\em Journal of High Energy Physics}, (3), 2022.
\newblock
  \texttt{\href{http://dx.doi.org/10.1007/JHEP03(2022)205}{DOI:\,10.1007/JHEP03(2022)205}}.

\bibitem{akers2022page}
Chris Akers, Thomas Faulkner, Simon Lin, and Pratik Rath.
\newblock {The Page curve for reflected entropy}.
\newblock {\em Journal of High Energy Physics}, 2022(6), 2022.
\newblock
  \texttt{\href{http://dx.doi.org/10.1007/JHEP06(2022)089}{DOI:\,10.1007/JHEP06(2022)089}}.

\bibitem{burri2024minreflected}
Laura Burri.
\newblock {Min-reflected entropy (to appear)}, 2024.

\bibitem{tomamichel2014relating}
Marco Tomamichel, Mario Berta, and Masahito Hayashi.
\newblock {Relating different quantum generalizations of the conditional
  R\'enyi entropy}.
\newblock {\em Journal of Mathematical Physics}, 55(8), 2014.
\newblock
  \texttt{\href{http://dx.doi.org/10.1063/1.4892761}{DOI:\,10.1063/1.4892761}}.

\bibitem{fuchs2004cryptographic}
Christopher~A. Fuchs and Jeroen van~de Graaf.
\newblock {Cryptographic Distinguishability Measures for Quantum Mechanical
  States}, 1997.
\newblock
  \texttt{\href{http://dx.doi.org/10.48550/arXiv.quant-ph/9712042}{DOI:\,10.48550/arXiv.quant-ph/9712042}}.

\bibitem{hiai1981sufficiency}
Fumio Hiai, Masanori Ohya, and Makoto Tsukada.
\newblock {Sufficiency, KMS condition and relative entropy in von Neumann
  algebras}.
\newblock {\em Pacific Journal of Mathematics}, 96:99--109, 1981.
\newblock
  \texttt{\href{http://dx.doi.org/10.2140/PJM.1981.96.99}{DOI:\,10.2140/PJM.1981.96.99}}.

\end{thebibliography}

\end{document}